\documentclass[journal]{IEEEtran}

\usepackage{cite}
\usepackage{amsmath,amssymb,amsfonts}
\usepackage{algorithm}
\usepackage{algpseudocode}
\usepackage{textcomp}

\usepackage{mathtools}

\usepackage{tabularx, booktabs}
\usepackage{multirow, makecell}

\usepackage{graphicx}
\graphicspath{{./Figures/}}

\usepackage{epstopdf}
\usepackage{subfigure}
\usepackage{array}
\usepackage{enumerate}
\usepackage{enumitem}
\usepackage{color}

\usepackage{tikz, pgfplots}
\usetikzlibrary{positioning,fit,calc,backgrounds}

\usepackage{amsmath, amsthm, amsfonts, amssymb}
\usepackage{mathtools}
\usepackage{graphicx}
\graphicspath{{./Figures/}}
\usepackage{subfigure}
\usepackage{booktabs, multirow, array}
\usepackage{bm} 
\usepackage{xcolor}
\usepackage{hyperref}

\definecolor{darkblue}{RGB}{0,0,180}  
\usepackage{tikz}
\usepackage{cancel}
\hypersetup{
  colorlinks=true,     
  linkcolor=darkblue,     
  citecolor=blue,     
  urlcolor=blue        
}

 \def\cB{\mathcal{B}}  \def\cD{\mathcal{D}}
\def\cE{\mathcal{E}} \def\cF{\mathcal{F}}  
   \def\cL{\mathcal{L}}
 \def\cN{\mathcal{N}} \def\cO{\mathcal{O}} \def\cP{\mathcal{P}}
   \def\cT{\mathcal{T}}
  \def\cW{\mathcal{W}}

 \def\vb{{\bf b}}  
 \def\vf{{\bf f}} \def\vg{{\bf g}} 
\def\vi{{\bf i}}   \def\vl{{\bf l}}
   \def\vp{{\bf p}}
\def\vq{{\bf q}} \def\vr{{\bf r}} \def\vs{{\bf s}} 
\def\vu{{\bf u}} \def\vv{{\bf v}} \def\vw{{\bf w}} \def\vx{{\bf x}}
\def\vy{{\bf y}} \def\vz{{\bf z}}

\def\vA{{\bf A}} \def\vB{{\bf B}}  \def\vD{{\bf D}}
  \def\vG{{\bf G}} \def\vH{{\bf H}}
\def\vI{{\bf I}}  \def\vK{{\bf K}} 
\def\vM{{\bf M}} \def\vN{{\bf N}}  
 \def\vR{{\bf R}} \def\vS{{\bf S}} \def\vT{{\bf T}}
\def\vU{{\bf U}} \def\vV{{\bf V}}  \def\vX{{\bf X}}
\def\vY{{\bf Y}} 

\def\R{\mathbb{R}}
\def\N{\mathbb{N}}
\def\C{\mathbb{C}}

\def\st{\text{s.t.}}

\newcolumntype{C}[1]{>{\centering\arraybackslash}m{#1}}

\newcommand{\std}[1]{{\fontsize{3.5}{5}\selectfont \ \textpm {#1}}}

\newtheorem{lemma}{Lemma}
\newtheorem{prop}{Proposition}

\newtheorem{assumption}{Assumption}

\newtheorem*{theorem*}{Theorem}

\usepackage{thmtools} 

\DeclareMathOperator*{\diag}{diag}

\DeclareMathOperator*{\argmin}{argmin}

\def\hvv{\widehat{\vv}}
\def\hvx{\widehat{\vx}}
\def\hvs{\widehat{\vs}}

\newcommand{\norm}[1]{\left\lVert#1\right\rVert_2}
\newcommand{\infnorm}[1]{\left\lVert#1\right\rVert_{\infty}}

\newcommand{\relu}[1]{\left[ #1 \right]_+}

\def\clip{\mathrm{clip}}

\definecolor{Blue1}{rgb}{0.3,0,1}

\begin{document}

\title{Scalable Self-Supervised Learning for \\ Multiphase AC-OPF in Distribution Systems with Topology Reconfiguration
}
\author{Hoang T. Nguyen$^1$, Shaohui Liu$^1$,
Reetam Sen Biswas$^2$, Varsha Pendyala$^2$, Nurali Virani$^2$,
Deepjyoti Deka$^1$, and Priya L. Donti$^1$%
\thanks{ 
	$^1$ Massachusetts Institute of Technology, Cambridge, MA, USA.}
\thanks{$^2$ GE Vernova Advanced Research Center, Niskayuna, NY, USA.}
}

\maketitle

\begin{abstract}
The proliferation of distributed energy resources (DERs) in distribution grids enables the active coordination of these assets to reduce costs and enable cleaner operations.
Realizing this potential requires solving multiphase AC optimal power flow (AC-OPF) quickly across varying loads, DER availabilities, and topology reconfigurations, at much greater speed and scale than conventional nonlinear solvers.
Learning-based surrogates can offer millisecond inference, yet existing methods target largely balanced transmission systems and do not scale to the multiphase, unbalanced, and reconfigurable nature of distribution feeders at utility scale. We present the Penalty + Sequential Linearized Feasibility Seeking (SLFS) algorithm, a self-supervised learning framework for multiphase distribution AC-OPF under switch-induced topology changes.
Penalty+SLFS requires no labeled optimal solutions and trains directly from the AC-OPF objective and constraints through a differentiable fixed-point power flow solver, avoiding expensive label generation and admitting robust training procedures. Topology changes are handled efficiently using Sherman–Morrison–Woodbury updates of the admittance-matrix inverse, while an $M$-step Jacobian approximation accelerates differentiation through the power flow solver. At inference, SLFS repairs any infeasible predictions, providing feasibility guarantees with low computational overhead.
On IEEE feeders ranging from 13 to 8{,}500 nodes, Penalty+SLFS achieves negligible optimality gaps and near-zero constraint violations, delivers up to three orders of magnitude speedups over IPOPT, and remains robust under large distributional shifts, demonstrating a viable path toward real-time, topology-aware AC-OPF for large-scale distribution grids.
\end{abstract}

\section{Introduction}

Distribution grids are evolving into actively controlled cyber--physical systems.
The growing penetration of inverter-interfaced distributed energy resources (DERs) and electrified loads intensifies variability and tightens operational limits such as voltage and current constraints~\cite{licari2025review}.
These challenges are exacerbated in distribution feeders, 
which are often multiphase and unbalanced, include switches inducing topology changes, and require operators to evaluate many load/DER/switch scenarios in near real time \cite{bazrafshan2017comprehensive,hao2024safe}.
AC optimal power flow (AC-OPF) provides a structured approach to coordinate distribution assets and minimize constraint violations.
Unfortunately, AC-OPF is inherently nonlinear and nonconvex, making it difficult to solve reliably and quickly in large-scale systems. 
These long solve times pose a bottleneck for real-time operations, look-ahead analysis, and interconnection studies.

Learning-based approaches can accelerate OPF by, e.g., learning warm starts~\cite{baker2019learning}, predicting active constraints~\cite{xavier2021learning}, learning iterative steps of a solver~\cite{ajeyemi2025learning}, or learning direct solution mappings~\cite{huang2021deepopf,dong2020smart,dontidc3,nguyen2025fsnet}.
Among these, the latter approach offers the potential for orders-of-magnitude speedups by directly mapping from operating conditions to OPF solutions,
often in milliseconds at inference time.
For instance, learning-based approaches for transmission AC-OPF are able to produce high-quality solutions two to three orders of magnitude faster than
traditional solvers~\cite{huang2021deepopf,dontidc3};
however, these gains have not yet been translated to distribution feeders, motivating our work.

The literature has explored two main approaches for learning direct solution mappings:
supervised learning and self-supervised learning.
Supervised learning approaches learn mappings by using offline-generated optimal solutions as labels within a regression error loss function~\cite{karagiannopoulos2019data,torre2022decentralized,singh2020learning}. 
While supervised training is typically straightforward, it requires labeled OPF solutions generated via successful AC-OPF solves; such labels can be expensive to obtain at scale and generally represent only feasible points (i.e., infeasible samples
are generally eliminated from training).
While dataset pipelines such as OPF-Learn~\cite{joswig2022opf} tighten sampling around the feasible region and increase the yield of usable labels, data generation can still take hours to days and scale poorly with grid size~\cite{klamkin2025PGLearn,park2024self,joswig2022opf}.
Recent grid foundation models follow this supervised, labeled-data route at transmission scale: IBM GridFM/GENCO unifies power flow, AC-OPF, and state estimation in one architecture~\cite{hamann2024foundation,puech2026genco}, Microsoft GridSFM predicts AC-OPF solutions across multiple transmission topologies~\cite{yang2026gridsfm}, and LUMINA evaluates multi-topology pretraining and transfer for AC-OPF surrogates~\cite{li2026lumina}.
Self-supervised approaches avoid this labeled-data bottleneck by training directly against the optimization structure, i.e., the objective and constraint violations~\cite{dontidc3,nguyen2025fsnet,pan2022deepopf,huang2021deepopf}. 
While this sometimes comes at the cost of training stability~\cite{nguyen2026cheap}, nonetheless, 
self-supervised methods for transmission AC-OPF and SCOPF
have been shown to match or exceed supervised surrogates while avoiding large labeled datasets~\cite{park2024self,pareek2025optimization,anrrango2026self}.

A key challenge for both self-supervised and supervised methods is feasibility enforcement, i.e., raw neural-network predictions frequently violate physical constraints.
To address this, many works penalize constraint violations in the loss via Lagrange multipliers~\cite{gupta2022dnn, park2023self,kim2023self,bouchkati2024augmented,li2026lumina}, which can improve reliability relative to pure supervised objectives, especially under distributional shift~\cite{li2026lumina}.
However, while these methods encourage constraint satisfaction, they do not strictly enforce it, so trained models may still return infeasible solutions~\cite{dontidc3}.
An alternative is to include hard feasibility enforcement components within training and inference, correcting raw predictions into feasible solutions~\cite{dontidc3,nguyen2025fsnet,pan2022deepopf}.
Although such methods can enforce feasibility, 
they are often substantially more expensive than pure penalty-based methods due to the cost of the enforcement process and the need to differentiate through it during training.

Distribution AC-OPF imposes three additional requirements that existing surrogates rarely satisfy simultaneously:~handling multiphase unbalanced physics, validity across discrete switch configurations, and feeders with thousands of nodes.
Switch reconfiguration is routine in distribution operations for loss minimization and service restoration, and a separate literature optimizes the switch positions themselves~\cite{hao2024safe,authier2024graphyr,qin2025physics}.
Because each configuration alters the nodal admittance matrix and the power flows, a surrogate must condition on the switch status.
Multiphase distribution learning has targeted power flow rather than OPF~\cite{ghamizi2025powerflowmultinet}, while supervised distribution OPF surrogates generally either assume a fixed topology and/or inherit the labeled-data cost discussed above~\cite{mahto2024gat}.

Motivated by these gaps, we contribute the following:

\begin{enumerate}
	\item \underline{Penalty+SLFS approach:} We introduce a self-supervised surrogate for multiphase distribution AC-OPF that includes an efficient differentiable power-flow solver in its architecture, trains with a penalty-based loss, and employs a novel feasibility-seeking procedure at inference. Penalty+SLFS scales to large unbalanced feeders with switch-induced topology reconfiguration and remains accurate and robust under shifts in loads and DERs.

	\item \underline{Scalable feasibility enforcement:} We propose a sequential linearized feasibility-seeking (SLFS) algorithm that iteratively refines setpoints via successive linearizations and exploits distribution network structure to enforce operational limits at inference time. The algorithm uses only matrix--vector operations and is GPU-friendly. Under standard assumptions, we prove that constraint violations decrease monotonically using SLFS. 
	\item \underline{Efficient training:} 
    We improve training efficiency via two innovations. (i) We use an $M$-step Jacobian approximation for unrolled differentiation through the fixed-point power flow solver, which improves speed and  memory-efficiency compared to implicit differentiation. We provide approximation error bounds under standard assumptions.
	(ii) To efficiently handle topology reconfiguration, we use the Sherman--Morrison--Woodbury (SMW) identity, with iterative refinement for ill-conditioned cases, to efficiently update the inverse admittance matrix across switch configurations in training and inference.
	\item \underline{Robust training:}
    Unlike supervised pipelines that rely on specialized feasible-region sampling, our method trains on both feasible and infeasible instances, 
    improving generalization and robustness under distributional shift.
    This is possible because training requires only input samples---not labeled OPF solutions---and thus data are inexpensive to create and need not be limited to feasible points. 
	\item \underline{Large-scale validation:} We validate Penalty+SLFS on multiphase distribution systems with topology reconfiguration across grids ranging in size from 13 to 8{,}500 nodes. Penalty+SLFS provides scalable, high-quality solutions with inference speedups of up to two to three orders of magnitude over IPOPT on medium-to-large feeders.
\end{enumerate}

\begin{table}[t]
	\vspace{-2.5em}
	\caption{Nomenclature}
    \vspace{-0.5em}
	\label{tab:nomenclature}
	\scriptsize
	\centering
	\begin{tabular}{ll}
	\hline
	\textbf{Symbol} & \textbf{Description} \\
	\hline
	$\mathcal{P} = \{1,\ldots,N_p\}$ & Non-slack phase nodes \\
	$\mathcal{E}, \mathcal{E}_{\text{sw}}$ & Static phase-level branches and switch branches \\
	$\mathcal{D}$ & Set of DER units \\
	$N_{\text{c}}$ & Number of DER unit--phase controls \\
	$n_k \in \mathcal{P}$ & Injection node of the $k$-th DER setpoint \\
	$\vz = (z_{mn})_{(m,n) \in \mathcal{E}_{\text{sw}}}$ & Switch status ($1$=closed, $0$=open) \\
	$\vv$ , $|\vv|^{\text{min/max}}$ & Complex voltages $(v_n)_{n \in \mathcal{P}}$ and magnitude bounds\\
	$\alpha_{mn}^{\text{max}}$ & Max.\ angle difference on $(m,n)$ \\
	$\vs^{\text{load}},\, \vs^{\text{der}}$ & Load/DER apparent-power injections \\
	$p_{j,\phi}^{\text{der}},\, q_{j,\phi}^{\text{der}}, s_{j,\phi}^{\text{der,max}}$ & Active/reactive/max. apparent power of DER unit $j$, phase $\phi$ \\
    $p_j^{\text{der}},\, p_j^{\text{der,avail}}$ & Aggregate and available active power of unit $j$ \\
	$\vp^\text{der},\, \vq^\text{der}$ & Stacked DER $P$/$Q$ setpoints (length $N_{\text{c}}$) \\
	$\vp^{\text{der,min/max}},\, \vq^{\text{der,min/max}}$ & Bounds on $\vp^{\text{der}}$, $\vq^{\text{der}}$ \\
	$s_{j,\phi}^{\text{der,max}}$ & Max.\ apparent power of unit $j$, phase $\phi$ \\
	$\vY(\vz)$ & Topology-dependent admittance matrix \\  
	$y_{mn}$ & Series admittance of branch $(m,n)$ \\
	$i_{mn},\, i_{mn}^{\text{max}}$ & Branch current and thermal rating on $(m,n)$ \\
	$\vs_0,\, \vs_0^{\text{max}}$ & Substation apparent power and its limit \\
    $\Phi_0$, $\Phi_j$ & Phases at the slack bus and of DER unit $j$ \\
	\hline
	\end{tabular}
	\vspace{-2.0em}
\end{table}

\noindent\textbf{Notation}: Bold lower-case (upper-case) letters denote vectors (matrices);	$(\cdot)^\top$ and $\diag(\cdot)$ denote transpose and diagonalization.
For a complex quantity, $\Re\{\cdot\}$, $\Im\{\cdot\}$, $\jmath$, and $\overline{(\cdot)}$ denote the real part, imaginary part, imaginary unit, and complex conjugate.
For $\vA \in \mathbb{C}^{n \times m}$, $\|\vA\|_p$ is the induced $p$-norm.
The map $\clip(\vx, \vl, \vu) = \min(\max(\vx, \vl), \vu)$ denotes element-wise clipping.
For $\vf:\mathbb{C}^n \to \mathbb{C}^m$, 
$J_{\vs}\vf(\vs)$ denotes the real 
representation of its Jacobian.

\vspace{-0.2em}

\section{AC-OPF Formulation}
We consider AC-OPF on a multiphase distribution feeder with switch-based topology reconfiguration.
The feeder connects to the main grid at a substation (slack) bus with fixed voltage $\vv_0$; all other phase nodes are PQ buses with loads and controllable DERs.
The nodal admittance matrix depends on the switch-status vector $\vz$ as follows \cite{bazrafshan2017comprehensive}:
\begin{align*}
	\vY(\vz) = \begin{bmatrix}
		\vY_{LL}(\vz) & \vY_{L0}(\vz) \\
		\vY_{0L}(\vz) & \vY_{00} \\
	\end{bmatrix}.
\end{align*}
Specifically, the AC-OPF problem is formulated as:
\begin{subequations}
	\begin{align}
		& \hspace{-0.6cm} \min_{\vp^\text{der}, \vq^\text{der}, \vv} \,\,  w_1 C_1(\vs_0, \vp^\text{der}) + w_2 C_2(\vp^\text{der}, \vp^\text{der,avail}) \label{eq:obj} \\
		\st \,\, 
		& \vi = \vY_{L0}(\vz) \vv_0 + \vY_{LL}(\vz) \vv, \label{eq:vi_mat}\\
		& \vs^{\text{der}} - \vs^{\text{load}} = \diag(\vv) \overline{\vi}, \label{eq:power_balance}\\
		& \vs_0 = \diag(\vv_0) (\overline{\vY}_{00} \overline{\vv}_0 + \overline{\vY}_{0L}(\vz) \overline{\vv}), \label{eq:vs0_mat}\\
		& |s_{0,\phi}| \leq s_{0,\phi}^{\text{max}}, \phi \in \Phi_0, \label{eq:s0_lim}\\
		& |\vv|^{\text{min}} \leq |\vv| \leq |\vv|^{\text{max}}, \label{eq:v_lim}\\
		& \vp^{\text{der,min}} \leq \vp^{\text{der}} \leq \vp^{\text{der,max}}, \label{eq:p_der_lim}\\
		& \vq^{\text{der,min}} \leq \vq^{\text{der}} \leq \vq^{\text{der,max}}, \label{eq:q_der_lim}\\
		& \sum_{\phi \in \Phi_j} p_{j,\phi}^{\text{der}} \leq p_j^{\text{der,avail}}, \quad \forall j \in \cD, \label{eq:p_der_avail}\\
		& (p_{j,\phi}^{\text{der}})^2 + (q_{j,\phi}^{\text{der}})^2 \leq (s_{j,\phi}^{\text{der,max}})^2, \; \forall j \in \cD, \phi \in \Phi_j, \label{eq:s_der_lim}\\
		& i_{mn} = y_{mn}\,(v_m - v_n), \quad \forall (m,n) \in \cE, \label{eq:i_def}\\
		& |i_{mn}| \leq i_{mn}^{\text{max}}, \quad \forall (m,n) \in \cE, \label{eq:i_lim}\\
		& |\Im\{\overline{v}_m v_n\}| \leq \tan(\alpha_{mn}^{\text{max}}) \Re\{\overline{v}_m v_n\}, \, \forall (m,n) \in \cE, \label{eq:angle_lim}
	\end{align}
	\label{prob:ac-opf}
\end{subequations}
where the operation cost is defined as $C_1 (\vs_0, \vp^\text{der}) = c_1 p_0^2 + c_2 p_0 + c_3 + \sum_{j \in \cD} (c_{1,j}^\text{der} (p_j^\text{der})^2 + c_{2,j}^\text{der} p_j^\text{der} + c_{3,j}^\text{der} )$ with $p_0 = \sum_{\phi \in \Phi_0} \Re\{s_{0,\phi}\}$ and $p_j^{\text{der}} = \sum_{\phi \in \Phi_j} p_{j,\phi}^{\text{der}}$, and curtailment cost is defined as $C_2(\vp^\text{der}, \vp^\text{der,avail}) = \sum_{j \in \cD} c_{4,j}^\text{der} (p_j^\text{der,avail} - p_j^\text{der})$, with positive cost coefficients $c_1, c_2, c_3, c_{1,j}^\text{der}, c_{2,j}^\text{der}, c_{3,j}^\text{der}, c_{4,j}^\text{der}$.
	
Constraint \eqref{eq:vi_mat} enforces Kirchhoff's current law at all phase nodes $n \in \cP$, and \eqref{eq:power_balance} enforces complex power balance at each PQ node; \eqref{eq:vs0_mat}--\eqref{eq:s0_lim} define and limit the per-phase substation injection; \eqref{eq:v_lim} bounds nodal voltage magnitudes; \eqref{eq:p_der_lim}--\eqref{eq:s_der_lim} constrain real/reactive setpoints, available power, and inverter capability at each DER; \eqref{eq:i_def}--\eqref{eq:i_lim} define branch currents and impose thermal limits on static branches $(m,n) \in \cE$; and \eqref{eq:angle_lim} bounds voltage-angle differences across $\cE$.

\section{Proposed Method}

\subsection{Method Overview}\label{sec:method_overview}
\begin{figure}
	\centering
	\includegraphics[width=0.98\linewidth]{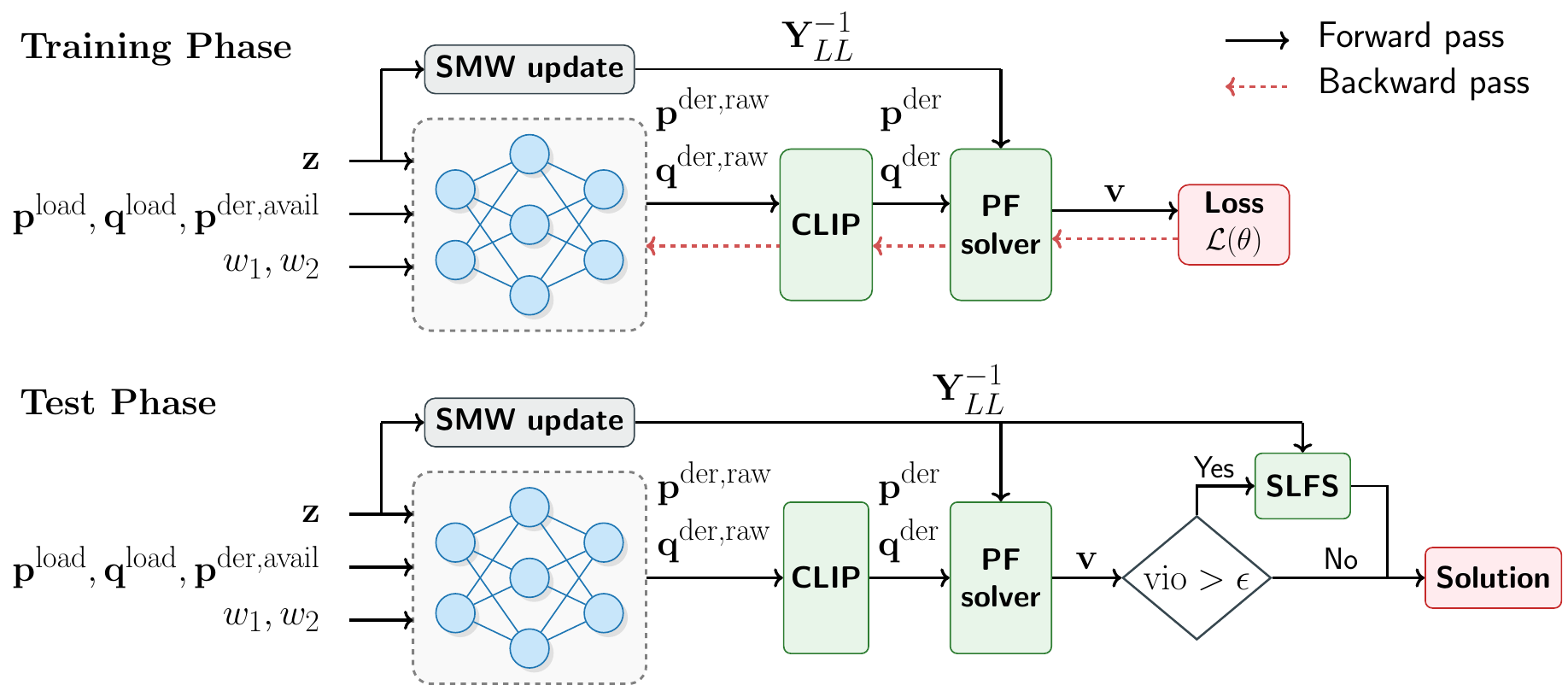}
    \setlength{\abovecaptionskip}{-2pt}
	\caption{Overview of the proposed training and inference pipelines.}
	\label{fig:penalty_fs_method}
    \vspace{-0.2in}
\end{figure}

We present Penalty+SLFS, a self-supervised
surrogate for multiphase distribution AC-OPF, depicted by 
Fig.~\ref{fig:penalty_fs_method}.
Given loads, DER availability, switch configuration $\vz$, and objective weights as input, the neural network outputs raw DER setpoints $(\vp^{\text{der,raw}}, \vq^{\text{der,raw}})$.
A clip layer (Sec.~\ref{sec:clip_layer}) then projects these outputs onto the DER operational limits to obtain the physically feasible setpoints $(\vp^{\text{der}}, \vq^{\text{der}})$.
A differentiable fixed-point power flow solver (Sec.~\ref{sec:power_flow_solver}) computes the corresponding voltages $\vv$, which are used to evaluate the AC-OPF objective and the remaining inequality constraints directly, without requiring labeled optimal solutions from a traditional AC-OPF solver.
The training loss combines this AC-OPF objective value with a penalty on the constraint violations, and the resulting gradient is backpropagated through the power flow solver and the neural network to update the network parameters in a self-supervised manner (Sec.~\ref{sec:loss}).
Because a change in switch configuration alters $\vY_{LL}(\vz)$ and its inverse, we update $\vY_{LL}^{-1}(\vz)$ efficiently using a SMW formula (Sec.~\ref{sec:smw}) rather than recomputing it from scratch.

At inference time, the same forward pass is used, except that sequential linearized feasibility seeking (SLFS, Sec.~\ref{sec:slfs}) is additionally applied to correct the predicted setpoints whenever the total constraint violation exceeds a threshold $\epsilon$.

\subsection{Clip Layer}\label{sec:clip_layer}
Let $\vx^\text{raw} \coloneq \bigl((\vp^{\text{der,raw}})^\top, (\vq^{\text{der,raw}})^\top\bigr)^\top$ denote the vector of raw DER setpoints produced by the neural network, and let $\vl_\vx \coloneq \bigl((\vp^{\text{der,min}})^\top, (\vq^{\text{der,min}})^\top\bigr)^\top$ and $\vu_\vx \coloneq \bigl((\vp^{\text{der,max}})^\top, (\vq^{\text{der,max}})^\top\bigr)^\top$ denote the corresponding box bounds.
The clip layer $\text{CLIP}(\vx^\text{raw})$ maps $\vx^\text{raw}$ to a setpoint $\vx \coloneqq \bigl((\vp^{\text{der}})^\top, (\vq^{\text{der}})^\top\bigr)^\top$ satisfying the DER limits \eqref{eq:p_der_lim}--\eqref{eq:s_der_lim} through the three sequential steps in Algorithm~\ref{alg:clip_layer}: 
(i) clipping to the box bounds $\vl_\vx, \vu_\vx$; 
(ii) uniformly scaling each unit's per-phase active power so that the aggregate active power $p_j^{\text{der}} = \sum_{\phi \in \Phi_j} p_{j,\phi}^{\text{der}}$ does not exceed the available power $p_j^{\text{der,avail}}$;
(iii) clipping each unit's reactive power so that the resulting setpoint satisfies the apparent-power capability \eqref{eq:s_der_lim}.

\begin{figure}[t]
\vspace{-0.15in}
\begin{algorithm}[H]
\caption{Clip layer $\text{CLIP}(\vx^\text{raw})$}
\label{alg:clip_layer}
\footnotesize
\begin{algorithmic}[1]
	\State $\bigl(\vp^{\text{der}}, \vq^{\text{der}}\bigr) \leftarrow \clip(\vx^\text{raw}, \vl_\vx, \vu_\vx)$ \Comment{{\scriptsize box limits}}
	\For{$j \in \cD$}
		\State $p_j^{\text{der}} \leftarrow \sum_{\phi \in \Phi_j} p_{j,\phi}^{\text{der}}$
		\State $\gamma_j \leftarrow \min\!\bigl(1,\, p_j^{\text{der,avail}}/p_j^{\text{der}}\bigr)$ if $p_j^{\text{der}} > 0$, else $1$
		\State $p_{j,\phi}^{\text{der}} \leftarrow \gamma_j\, p_{j,\phi}^{\text{der}}, \quad \forall \phi \in \Phi_j$ \Comment{{\scriptsize available-power scaling}}
		\State $q_{j,\phi}^{\text{cap}} \leftarrow \sqrt{(s_{j,\phi}^{\text{der,max}})^2 - (p_{j,\phi}^{\text{der}})^2}$
		\State $q_{j,\phi}^{\text{der}} \leftarrow \clip\bigl(q_{j,\phi}^{\text{der}},\, -q_{j,\phi}^{\text{cap}},\, q_{j,\phi}^{\text{cap}}\bigr), \quad \forall \phi \in \Phi_j$
		\hfill \Comment{{\scriptsize apparent-power limit}}
	\EndFor
	\State \Return $(\vp^{\text{der}}, \vq^{\text{der}})$
\end{algorithmic}
\end{algorithm}
\vspace{-0.3in}
\end{figure}

\subsection{Fixed-point based Power Flow Solver and Differentiation}\label{sec:power_flow_solver}

After the clip layer, evaluating the AC-OPF objective and the remaining inequality constraints (substation, voltage, and line limits) requires obtaining the bus voltages induced by the predicted setpoints.
As in other self-supervised OPF surrogates~\cite{dontidc3,pan2022deepopf}, we therefore embed a power flow solver in the forward pass and use its outputs to compute constraint violations that enter the training loss.
For the forward solve, we use fixed-point iteration, which is often faster than Newton--Raphson on distribution networks~\cite{dugan2011Open}.
The power flow equations \eqref{eq:vi_mat}--\eqref{eq:power_balance} can be written in the fixed-point form:
\begin{align}\label{eq:fixed-point}
    \vv = \vw(\vz) + \vY^{-1}_{LL}(\vz) \diag \left(\overline{\vv} \right)^{-1} \overline{\vs} =: \vG(\vv, \vs; \vz),
\end{align}
where $\vw(\vz) \coloneqq -\vY_{LL}^{-1}(\vz)\vY_{L0}(\vz)\vv_0$ is the no-load voltage vector and and $\vs = \vs^{\text{der}} - \vs^{\text{load}}$. 
Starting with an initial guess $\vv_{0}$, we iteratively update
\begin{align}\label{eq:fp_iter}
    \vv_{k+1} = \vG(\vv_{k}, \vs; \vz) 
\end{align}
until $\|\vv_{k+1} - \vv_{k}\|_\infty$ is below a predefined threshold. The convergence, uniqueness, and existence of the solution of \eqref{eq:fp_iter} have been well-established in \cite{wang2018Explicit,bernstein2018Loadb,bazrafshan2018Convergence}.

To backpropagate through the solver, either implicit or unrolled differentiation can be used~\cite{mandi2024DecisionFocused}.
Implicit differentiation requires a linear solve in the Jacobian of $\vG$ at the fixed point; because $\vG$ is not holomorphic, this system must be split into real and imaginary parts, doubling its size and becoming costly at scale.
We therefore unroll the iterations and truncate the backward pass.
Inspired by the one-step differentiation method in \cite{bolte2023one}, our $M$-step Jacobian approximation differentiates only the last $M$ iterations of \eqref{eq:fp_iter} via automatic differentiation, which is easy to implement and sufficiently accurate for training.

Formally, we define the $k$-fold composition of the power flow mapping $\psi_k(\vv,\vs) \coloneqq \vG^{\circ k}(\vv, \vs)$.
For a given power injection $\vs$, let $\vv_{j}(\vs)$ denote the voltage obtained after $j$ iterations of \eqref{eq:fp_iter} and $\vv_\star(\vs)$ denote the corresponding fixed point. The $M$-step approximate Jacobian is defined as
\begin{align*}
    J^{(M)}_\vs \vv_{k}(\vs) =&\; J_\vs \psi_M(\vv_{k-M}, \vs). 
\end{align*}
The approximation is obtained by dropping the sensitivity with respect to $\vv$ in the chain rule and treating $\vv_{k-M}(\vs)$ as constant.

To characterize the error of the $M$-step Jacobian approximation, we make the following assumption.
\begin{assumption}\label{as:operating_domain}
Let $\vR(\vz) \coloneqq \vY_{LL}^{-1}(\vz)$.
Define the operating domain
\begin{align*}
    \Omega 
    = \{ (\vv, \vs) : \|\vs\|_\infty \leq s^\text{max}, \ \min_i |v_i| \geq v^\text{min} \},
\end{align*}
where we assume that $v^\text{min} > 0$ and $s^\text{max} < \infty$. 
For each admissible switch configuration $\vz$, assume further that
\[
    \frac{\|\vR(\vz)\|_\infty \, s^\text{max}}{(v^\text{min})^2} < 1.
\]
\end{assumption}
\noindent The specified voltage and power bounds hold under normal operating conditions.
The final inequality is then readily computable and ensures that the fixed-point mapping $\vG(\cdot, \vs; \vz)$ is a contraction for all $\vs$ in the specified range~\cite{bernstein2018Loadb,bazrafshan2018Convergence,wang2018Explicit}. 

Under this assumption, $\vG$ and its Jacobian satisfy the following local Lipschitz bounds, which we use to establish the approximation error of the $M$-step Jacobian.
\begin{restatable}[Local Lipschitz continuity]{lemma}{Continuity}\label{lem:continuity}
    Suppose Assumption~\ref{as:operating_domain} holds. Define 
    \begin{align*}
        \beta \coloneqq \frac{\infnorm{\vR(\vz)} s^\text{max}}{(v^\text{min})^2}, 
        L_\vs \coloneqq \frac{\infnorm{\vR(\vz)}}{v^\text{min}},
        L_\vv \coloneqq \frac{2 \infnorm{\vR(\vz)} s^\text{max}}{(v^\text{min})^3}.
    \end{align*}
    Then, for all $(\vv, \vs), (\vv_1, \vs), (\vv_2, \vs) \in \Omega$, the following hold:
    \begin{align*}
        &\infnorm{J_\vv \vG(\vv, \vs; \vz)} \leq \beta, \quad \infnorm{J_\vs \vG(\vv, \vs; \vz)} \leq L_\vs, \\
        & \infnorm{J_\vv \vG(\vv_1, \vs; \vz) - J_\vv \vG(\vv_2, \vs; \vz)} \leq L_\vv \infnorm{\vv_1 - \vv_2}.
    \end{align*}
\end{restatable}

Let $\Delta \vv_{k}(\vs) \coloneqq \vv_{k}(\vs) - \vv_\star(\vs)$ denote the discrepancy between the $k$-th iteration and the fixed point.

\begin{prop}[Approximation error] \label{prop:approx_error}
    Under Assumption~\ref{as:operating_domain} and for any integers 
    $0 \leq M \leq k$, the following bound holds:
    \begin{align*}
         &\infnorm{J^{(M)}_\vs \vv_{k}(\vs) - J_\vs \vv_\star(\vs)} \\
         &\qquad\qquad \leq  \dfrac{\beta^M}{1 - \beta} L_\vs + \sum_{j=1}^{M} L_j \beta^{j - 1} \infnorm{\Delta \vv_{k - j}(\vs)}
    \end{align*}
    where the constants $\{L_j\}_{j=1}^M$ are defined as
    \begin{align*}
        L_j :=
        \begin{cases}
        L_{\vs} + \dfrac{L_{\vv} L_{\vs}}{1-\beta}, & \text{if } j < M,\\
        L_{\vs}, & \text{if } j = M .
        \end{cases}
    \end{align*}
\end{prop}
For sufficiently large $k$, such that $\vv_k(\vs)$ is close to the fixed point $\vv_\star(\vs)$, the second term in the bound of Proposition~\ref{prop:approx_error} becomes negligible. Thus, the approximation error of the $M$-step Jacobian approximation decays exponentially with $M$.

\subsection{Admittance Inverse via SMW}\label{sec:smw}
When the switch configuration changes, the admittance matrix and its inverse need to be recomputed. However, directly computing a dense inverse at every topology sample is costly when $|\cP|$ is large.
We therefore choose a fixed configuration $\vz^0$, and update $\vY_{LL}^{-1}(\vz)$ from the reference using a SMW formula, which we show empirically to be substantially faster than direct inversion as the grid size grows (Sec.~\ref{sec:smw_results}).

With a reference switch vector $\vz^0$, we let $\vR_0 \coloneqq \vY_{LL}^{-1}(\vz^0) = \vR(\vz^0)$ and precompute it once offline.
Rather than rebuilding $\vY_{LL}(\vz)$ from scratch, we add or remove only the switch contributions that differ from $\vz^0$.
For each switchable branch $(m,n) \in \cE_{\text{sw}}$, define the signed status change
\begin{equation*}
\delta_{mn} \coloneqq z_{mn} - z_{mn}^0 \in \{-1,0,+1\},
\end{equation*}
where $\delta_{mn}=+1$ means $(m,n)$ closes relative to $\vz^0$, $\delta_{mn}=-1$ means it opens, and $\delta_{mn}=0$ means it is unchanged.

Let $\vS_{mn} \in \C^{|\cP| \times |\cP|}$ denote the contribution that branch $(m,n)$ adds to the load-block admittance matrix when it is closed ($z_{mn}=1$).
Let $\Phi_{mn}$ be the set of phases present on $(m,n)$, and let $n_{f,r}, n_{t,r} \in \cP$ denote the from-side and to-side phase-node indices for phase $r \in \Phi_{mn}$.
We build $\vS_{mn}$ from the switch $\pi$-model: for every pair of phases $(r,c) \in \Phi_{mn} \times \Phi_{mn}$, the series and shunt admittances $y^s_{rc}$ and $y^\sigma_{rc}$ are stamped into the four from/to blocks as
\begin{align*}
	[\vS_{mn}]_{n_{f,r}, n_{f,c}},\; [\vS_{mn}]_{n_{t,r}, n_{t,c}} &\mathrel{+}= y^s_{rc} + \tfrac{1}{2} y^\sigma_{rc}, \\
	[\vS_{mn}]_{n_{f,r}, n_{t,c}},\; [\vS_{mn}]_{n_{t,r}, n_{f,c}} &\mathrel{+}= -y^s_{rc},
\end{align*}
i.e., series and half-shunt terms on the from--from and to--to blocks, and $-y^s_{rc}$ on the cross blocks.
All other entries of $\vS_{mn}$ are zero, so $\vS_{mn}$ is supported only on the $2|\Phi_{mn}|$ rows and columns indexed by $\{n_{f,r}, n_{t,r}\}_{r \in \Phi_{mn}}$. Consequently, $\text{rank}(\vS_{mn}) \leq 2|\Phi_{mn}| \ll |\cP|$, which is what makes the Woodbury factorization below possible.
The load-block perturbation is then the weighted sum of these stamps:
\begin{equation*}
\Delta \vY_{LL}(\vz) = \sum_{(m,n) \in \cE_{\text{sw}}} \delta_{mn}\,\vS_{mn},
\end{equation*}
so $\vY_{LL}(\vz) = \vY_{LL}(\vz^0) + \Delta \vY_{LL}(\vz)$.
To apply the Woodbury update, each active term $\delta_{mn}\vS_{mn}$ is factored as $\vU_{mn}\vV_{mn}^\top$ with $\vU_{mn}, \vV_{mn} \in \C^{|\cP| \times 2|\Phi_{mn}|}$ nonzero only on the rows $\{n_{f,r}, n_{t,r}\}_{r \in \Phi_{mn}}$. The columns from all switches with $\delta_{mn}\neq 0$ are then stacked into $\vU(\vz), \vV(\vz) \in \C^{|\cP| \times n_{\text{upd}}}$, yielding
\begin{equation*}
\Delta \vY_{LL}(\vz) = \vU(\vz)\,\vV(\vz)^\top.
\end{equation*}
Here $n_{\text{upd}} = \sum_{(m,n):\,\delta_{mn}\neq 0} 2|\Phi_{mn}|$ is the total number of stacked factor columns and is typically much smaller than $|\cP|$, since only a few switches change per topology sample.

Because $\Delta \vY_{LL}(\vz)=\vU(\vz)\vV(\vz)^\top$ is low rank and $\vR_0$ is already available, we update the inverse with the SMW identity applied to $\vY_{LL}(\vz)=\vY_{LL}(\vz^0)+\vU(\vz)\vV(\vz)^\top$:
\begin{align}\label{eq:woodbury}
    \vY_{LL}(\vz)^{-1}
    = \vR_0
    - \vR_0\vU(\vz)
    \bigl(\vK(\vz)\bigr)^{-1}
    \vV(\vz)^\top \vR_0,
\end{align}
where only the small matrix
\begin{align*}
    \vK(\vz) \coloneqq \vI_{n_{\text{upd}}} + \vV(\vz)^\top \vR_0 \vU(\vz) \in \C^{n_{\text{upd}} \times n_{\text{upd}}}
\end{align*}
is inverted explicitly. The remaining work is matrix multiplication with the precomputed $\vR_0$.
Thus, each topology update costs $\cO(n_{\text{upd}}^3 + |\cP|^2 n_{\text{upd}})$, compared with $\cO(|\cP|^3)$ for a dense inversion of $\vY_{LL}(\vz)$.
When $n_{\text{upd}} \ll |\cP|$, this is a significant computational saving.

\textbf{Iterative refinement.}
Although \eqref{eq:woodbury} is exact analytically, when evaluated numerically, roundoff in the $n_{\text{upd}}\times n_{\text{upd}}$ solve for $\vK(\vz)^{-1}$ can be amplified by the subsequent multiplications with $\vR_0\vU(\vz)$ and $\vV(\vz)^\top\vR_0$, and the computed inverse may deviate from $\vY_{LL}(\vz)^{-1}$ \cite{hager1989updating}.
This effect is strongest when $\kappa(\vK(\vz))$ is large, e.g., under ill-conditioned switch reconfigurations with high-impedance tie switches or near-parallel paths where $\vU(\vz)$ and $\vV(\vz)$ have nearly dependent columns.

To correct for this issue without ever forming a dense inverse, let $\cW(\vz)[\cdot]$ denote the linear map given by the right-hand side of \eqref{eq:woodbury}, i.e., $\cW(\vz)[\vM] \coloneqq \vR_0\vM - \vR_0\vU(\vz)\bigl(\vK(\vz)\bigr)^{-1}\vV(\vz)^\top \vR_0\vM$.
Let $\vX_0 \coloneqq \cW(\vz)[\vI_{|\cP|}]$ be the SMW update of $\vY_{LL}(\vz)^{-1}$ from \eqref{eq:woodbury}.
When the update is ill-conditioned, we refine $\vX_0$ with one or two SMW iterative-refinement (SMW-IR) steps \cite{hashemi2025instability},
\begin{align*}
    \vX_{j+1} = \vX_j + \cW(\vz)\bigl[\vI_{|\cP|} - \vY_{LL}(\vz)\vX_j\bigr],
    \quad j = 0,1,2.
\end{align*}
The refined matrix $\vX_j$ is then used in place of $\vY_{LL}(\vz)^{-1}$ in the fixed-point power flow \eqref{eq:fixed-point}.

\subsection{Loss Function} \label{sec:loss}
After the clip layer and the power-flow solve, each training instance yields a predicted operating point from which we evaluate the AC-OPF cost and any remaining constraint violations.
We combine these into a penalty-based loss that drives the network toward low-cost, near-feasible solutions.

Let $\cB$ denote a training minibatch.
For each $b \in \cB$, let $\vx^{(b)}$ denote the DER setpoints produced by the clip layer for sample $b$, and let $\vv^{(b)}$ denote the corresponding voltage solution from the power flow solver.
We write $f(\vx^{(b)})$ for the objective \eqref{eq:obj} evaluated at this operating point, including the substation power implied by $\vv^{(b)}$, and $\vg(\vx^{(b)})$ for the stacked network operational constraints with associated bounds $\vl \le \vg(\vx) \le \vu$ (defined explicitly in Appendix~\ref{appendix:lin_constraints}), comprising the voltage, branch-current, angle-difference, and substation limits \eqref{eq:s0_lim}, \eqref{eq:v_lim}, \eqref{eq:i_lim}--\eqref{eq:angle_lim}.
Note that DER box, availability, and apparent-power limits \eqref{eq:p_der_lim}--\eqref{eq:s_der_lim} are enforced by the clip layer and are therefore excluded from $\vg(\cdot)$.
Power-flow equalities \eqref{eq:vi_mat}--\eqref{eq:power_balance} are imposed by the forward solver. 
The training loss is
\begin{align}
    \cL(\theta) &= \frac{1}{|\cB|} \sum_{b \in \cB} \Bigl( f(\vx^{(b)}) + \rho\, c(\vx^{(b)}) \Bigr), \label{eq:training_loss}\\
    c(\vx^{(b)}) &= \big\| [\vg(\vx^{(b)}) - \vu]_+ + [\vl - \vg(\vx^{(b)})]_+ \big\|_2^2,\nonumber
\end{align}
where $\rho > 0$ is a fixed weight, $[\cdot]_+ \coloneqq \max(\cdot,0)$ is applied elementwise, and $c(\vx^{(b)})$ is the penalty on inequality violations.

Prior work updates $\rho$ in a Lagrangian fashion during training~\cite{kim2023self, chen2022unsupervised}.
In our setting, however, training minibatches often contain infeasible instances with nonzero violations, so such updates tend to increase $\rho$ monotonically and skew the loss toward feasibility at the expense of cost. 
We therefore keep $\rho$ fixed, which also simplifies hyperparameter selection.

\subsection{Sequential Linearized Feasibility Seeking (SLFS)}\label{sec:slfs}
At inference, if residual violations of the network prediction exceed a tolerance $\epsilon$, SLFS iteratively repairs the DER setpoints (the bottom diagram of Fig.~\ref{fig:penalty_fs_method}).
Each iteration linearizes the nonlinear network constraints around the current operating point and solves a damped proximal feasibility subproblem on that linear model.
We derive the linearized model next, then present the algorithm and its convergence guarantees.

\subsubsection{Linearized model} \label{sec:linearized model}
Let $(\hvx, \hvv, \hvs)$ denote a reference operating point at the current SLFS iterate, with $\hvs = \hvs^\text{der} - \vs^\text{load}$ and loads $\vs^\text{load}$ held fixed. In the iteration described below, this reference point is the current iterate, $\hvx = \vx_k$.
Throughout this subsection, hat notation denotes quantities evaluated at the current reference $(\hvx,\hvv,\hvs)$. All linearization matrices below---including $\vB$, $\vA_v(\hvx)$, and the stacked sensitivity matrix $\vA(\hvx)$---are formed from this reference and recomputed whenever SLFS advances to a new iterate.

Recall the fixed-point relation \eqref{eq:fixed-point}, obtained by combining \eqref{eq:vi_mat} and \eqref{eq:power_balance}: $\vv = \vw(\vz) + \vY_{LL}^{-1}(\vz) \diag{(\overline{\vv})}^{-1} \overline{\vs}$, where $\vw(\vz)$ is the no-load voltage vector.
Following the fixed-point linearization in \cite{bernstein2018Loadb}, we freeze the voltage dependence in $\diag(\overline{\vv})^{-1}$ at the reference $\hvv$:
\begin{align*}
	\vv \approx \vw(\vz) + \vH\, \overline{\vs}, \qquad \vH \coloneqq \vY_{LL}^{-1}(\vz)\diag{(\overline{\hvv})}^{-1} \in \C^{N_p \times N_p}.
\end{align*}
Since $\hvv$ is itself the fixed point at the reference injection $\hvs$, it satisfies $\hvv = \vw(\vz) + \vH\overline{\hvs}$, i.e., $\vw(\vz) = \hvv - \vH\overline{\hvs}$. Substituting back gives the affine voltage approximation
\begin{align*}
	\vv \approx \hvv + \vH\,(\overline{\vs} - \overline{\hvs}).
\end{align*}
Because loads are fixed, only DER setpoints vary: a change $\delta p_k$ in the $k$-th active-power control injects $\delta p_k$ at phase node $n_k \in \cP$, and a change $\delta q_k$ in the corresponding reactive control changes $\overline{s_{n_k}}$ by $-\jmath\,\delta q_k$.
Stacking these local sensitivities yields the affine voltage model
\begin{align}
	\vv \approx \hvv + \vB (\vx - \hvx),
	\label{eq:lin_pf_simple}
\end{align}
where $\vB \coloneqq \begin{bmatrix} \vB_p & \vB_q \end{bmatrix} \in \C^{N_p \times 2N_{\text{c}}}$ has columns $(\vB_p)_{:,k} = \vH_{:,n_k}$ and $(\vB_q)_{:,k} = -\jmath\,\vH_{:,n_k}$ for $k = 1,\ldots,N_{\text{c}}$.

\textbf{Voltage magnitude.}
The voltage limits \eqref{eq:v_lim} depend on the per-node magnitudes $|v_n|$.
Differentiating $|v_n|$ with respect to the real decision vector $\vx$ and using \eqref{eq:lin_pf_simple} yields
\begin{align}
    |\vv| &\approx |\hvv| + \vA_v(\hvx) (\vx - \hvx), \label{eq:vmag_matrix} \\
    \vA_v(\hvx) &\coloneqq \diag(|\hvv|)^{-1} \Re\left\{\diag(\overline{\hvv})\vB \right\} \in \R^{N_p \times 2N_{\text{c}}}.
    \label{eq:Av_vmag}
\end{align}

The same magnitude and affine linearizations are applied to the substation-power, angle-difference, and branch-current limits \eqref{eq:s0_lim}, \eqref{eq:angle_lim}, and \eqref{eq:i_lim}; full derivations are given in Appendix~\ref{appendix:lin_constraints}.
In summary, the linearized network inequality limits are given by
\begin{align*}
	&\vl \leq \vA(\hvx) \vx + \vb(\hvx) \leq \vu, \\
	&\vl_\vx \leq \vx \leq \vu_\vx, \\
	&\eqref{eq:p_der_avail}, \eqref{eq:s_der_lim}.
\end{align*}
See Appendix~\ref{appendix:lin_constraints} for definitions of $\vA(\hvx)$, $\vb(\hvx)$, $\vl$, and $\vu$.

\subsubsection{Algorithm and Convergence Analysis}
The main idea of SLFS is to iteratively linearize the nonlinear constraints around the current operating point and solve a convex feasibility subproblem, repeating this for a fixed number of iterations.

We define the $\ell_\infty$-norm violation functions for the nonlinear and linearized constraints, respectively, as
\begin{align*}
    V(\vx)& = \infnorm{\relu{\vg(\vx) - \vu} + \relu{\vl - \vg(\vx)}},\\
    V^\text{lin} (\vx; \hvx)& = \|\relu{\vA(\hvx) \vx + \vb(\hvx) - \vu} \\
    &\qquad\;\;\; + \relu{\vl - \vA(\hvx) \vx - \vb(\hvx)} \|_\infty,
\end{align*}
and let the feasible set capturing the remaining DER-side constraints be
\begin{align*}
    \cF = \{ \vx : \vl_\vx \leq \vx \leq \vu_\vx, \eqref{eq:p_der_avail}, \eqref{eq:s_der_lim}\}.
\end{align*}
Note that $\cF$ collects exactly the box, availability, and apparent-power constraints \eqref{eq:p_der_lim}--\eqref{eq:s_der_lim} that $\text{CLIP}$ (Algorithm~\ref{alg:clip_layer}) enforces by construction. Hence, $\text{CLIP}$ maps any point into $\cF$, and we reuse it as the feasibility map onto $\cF$ inside Algorithm~\ref{alg:pdhg-fs}.

At iteration $k$, given the current iterate $\vx_k$, we solve the proximal feasibility problem
\begin{align}\label{eq:proximal}
    \tilde{\vx}_{k+1} = \argmin_{\vx \in \cF} V^\text{lin} (\vx; \vx_{k}) + \frac{\eta}{2} \norm{\vx - \vx_{k}}^2,
\end{align}
where $\eta > 0$. The proximal term $\frac{\eta}{2} \norm{\vx - \vx_{k}}^2$ regularizes the step and ensures the problem is well-defined. Here $V^\text{lin}(\cdot\,;\vx_k)$ is convex, being an $\ell_\infty$ norm of nonnegative $\relu{\cdot}$ applied to affine maps, and $\cF$ is convex. Adding the proximal term therefore makes the objective strongly convex with modulus $\eta$, so \eqref{eq:proximal} has a unique solution.

After obtaining $\tilde{\vx}_{k+1}$ from \eqref{eq:proximal}, we update the next iterate $\vx_{k+1}$ via a damping step as follows:
\begin{align} \label{eq:damping}
    \vx_{k+1} = \vx_{k} + \alpha_k \Delta \vx_k,
\end{align}
where $\Delta \vx_k \coloneqq \tilde{\vx}_{k+1} - \vx_{k}$.
The full SLFS procedure is summarized in Algorithm~\ref{alg:slfs}. The proximal subproblem \eqref{eq:proximal} is solved by primal-dual hybrid gradient (PDHG)~\cite{chambolle2011first} in Algorithm~\ref{alg:pdhg-fs}. Each PDHG iteration consists of a dual update and a primal update that are matrix-vector products with $\vA(\vx_k)$ and $\vA(\vx_k)^\top$, followed by the closed-form elementwise clip and the $\text{CLIP}$ feasibility map onto $\cF$ (Algorithm~\ref{alg:clip_layer}). Consequently, the algorithm requires no linear solves or Jacobian materialization and is efficient to implement on GPUs.
\begin{figure}[t]
\vspace{-0.2in}
\begin{algorithm}[H]
    \caption{Sequential linearized feasibility seeking}
    \label{alg:slfs}
    \footnotesize
    \begin{algorithmic}[1]
    \Require Prediction $(\vp^{\text{der}}, \vq^{\text{der}})$, load $\vs^{\text{load}}$, topology $\vz$, regularization $\eta > 0$, damping factors $\{\alpha_k\}$, max linearization iterations $K^{\text{lin}}$.
    
    \Ensure Approximately feasible DER setpoints $(\vp^{\text{der}}, \vq^{\text{der}})$.
    
    \State \textbf{Initialize:} $\vx_0 \coloneqq \bigl((\vp^{\text{der}})^\top, (\vq^{\text{der}})^\top\bigr)^\top$
    \State Construct $\vs^{\text{der}}$ from $(\vp^{\text{der}}, \vq^{\text{der}})$; $\vs \gets \vs^{\text{der}} - \vs^{\text{load}}$
    \State Solve $\vv_0 = \vG(\vv_0, \vs; \vz)$ via \eqref{eq:fp_iter}
    
    \For{$k = 0,1,\dots,K^{\text{lin}}-1$}
        \State \textbf{Linearize:} Construct $\vA(\vx_k), \vb(\vx_k)$ from $(\vv_k, \vx_k)$
        \State \textbf{Proximal step:} $\tilde{\vx}_{k+1} \gets$ Alg.~\ref{alg:pdhg-fs} with $(\vA(\vx_k), \vb(\vx_k), \vl, \vu, \vx_k, \eta)$
        \State \textbf{Damping:} $\Delta \vx_k \gets \tilde{\vx}_{k+1} - \vx_k$; $\vx_{k+1} \gets \vx_k + \alpha_k \Delta \vx_k$
        \State Update $\vs^{\text{der}}$ from $\vx_{k+1}$; $\vs \gets \vs^{\text{der}} - \vs^{\text{load}}$
        \State Solve $\vv_{k+1} = \vG(\vv_{k+1}, \vs; \vz)$ via \eqref{eq:fp_iter}
    \EndFor
    \State Extract $(\vp^{\text{der}}, \vq^{\text{der}})$ from $\vx_{K^{\text{lin}}}$
    \State \Return $\vp^{\text{der}}, \vq^{\text{der}}$
    \end{algorithmic}
\end{algorithm}
\vspace{-0.3in}
\end{figure}
\begin{figure}[t]
\vspace{-0.2in}
\begin{algorithm}[H]
    \caption{Primal-dual hybrid gradient method}
    \label{alg:pdhg-fs}
    \footnotesize
    \begin{algorithmic}[1]
    \Require Linearization $(\vA(\vx_k), \vb(\vx_k))$, bounds $(\vl, \vu)$, anchor $\vx_k$, tolerance $\epsilon$, max iterations $T$, regularization $\eta > 0$

    \Ensure Approximately feasible point $\vx \in \cF$

    \State \textbf{Feasibility map}: $\mathrm{CLIP}$ as in Algorithm~\ref{alg:clip_layer}, which maps any point into $\cF$
    \State \textbf{Parameters}: $L \approx \|\vA(\vx_k)\|_2$ (via power iteration)
    \State $\tau = \sigma = 0.95 / (L + 10^{-12})$, $\theta = 1$
    \State $\vx_0 \gets \mathrm{CLIP}(\vx_k)$, $\vy_0 \gets \mathbf{0}$, $\bar{\vx}_0 \gets \vx_0$
    \State $\vl \gets \vl - \vb(\vx_k)$, $\vu \gets \vu - \vb(\vx_k)$

    \For{$t=0,1,\dots,T-1$}
        \State $\tilde{\vy} \gets \vy_t + \sigma \vA(\vx_k) \bar{\vx}_t$
        \State $\vy_{t+1} \gets \tilde{\vy} - \sigma \cdot \clip(\tilde{\vy}/\sigma, \vl, \vu)$
        \State $\vx_{t+1} \gets \mathrm{CLIP}\!\left(\dfrac{\vx_{t} - \tau \vA(\vx_k)^\top \vy_{t+1} + \tau \eta \vx_k}{1 + \tau \eta}\right)$
        \State $\vr \gets \vA(\vx_k) \vx_{t+1} - \clip(\vA(\vx_k) \vx_{t+1}, \vl, \vu)$
        \If{$\|\vr\| < \epsilon$ \textbf{and} $\|\vx_{t+1} - \vx_{t}\|_2 < \epsilon$}
            \State \textbf{break}
        \EndIf
        \State $\bar{\vx}_{t+1} \gets \vx_{t+1} + \theta (\vx_{t+1} - \vx_{t})$
    \EndFor
    \State \Return $\vx_{t+1}$
    \end{algorithmic}
\end{algorithm}
\vspace{-0.3in}
\end{figure}

To analyze the convergence of SLFS, we make the following standard assumption on the linear model error.
\begin{assumption} \label{as:linear_model}
    For any $\hvx, \vx$ in the domain of interest, for some constants $\delta > 0$ and $\nu \in (0,1]$, the linear model satisfies:
    \begin{align*}
        & \infnorm{\vA(\hvx)\vx + \vb(\hvx) - \vg(\vx)} \leq \delta \infnorm{\vx - \hvx}^{1 + \nu},\\
        & \vA(\hvx) \hvx + \vb(\hvx) = \vg(\hvx).
    \end{align*}
\end{assumption}

\begin{restatable}{theorem}{ViolationDescentOne} \label{thm:violation_descent1}
    Suppose Assumption~\ref{as:linear_model} holds, and the damping factor $\alpha_k$ is chosen as
    \begin{align*}
        \alpha_k = \min\left(1, \left(\frac{c \eta}{2 \delta}\right)^{\frac{1}{\nu}} \norm{\Delta \vx_k}^{\frac{1 - \nu}{\nu}}\right)
    \end{align*}
    for some constant $c \in (0,1)$. Then, the sequence $\{\vx_k\}$ generated by \eqref{eq:proximal} and \eqref{eq:damping} satisfies
    \begin{align*}
        V(\vx_{k+1}) \leq V(\vx_k) - (1- c)\alpha_k \frac{\eta}{2} \norm{\Delta \vx_k}^2.
    \end{align*}
\end{restatable}
Thus, as long as the proximal step moves the iterate ($\Delta \vx_k \neq 0$), the violation decreases monotonically.

\section{Experimental Results}
\subsection{Experiment Setup}
We test our method on the IEEE~13-bus, 123-bus, 240-bus, 906-bus, and 8500-node feeders, which contain 3, 13, 22, 12, and 100~PV units and 2, 8, 9, 8, and 45~switches, respectively.

\textbf{Data generation.}
For every feeder, the training set contains 200{,}000 samples (both feasible and infeasible), and each test set contains 2{,}000 samples.
Training data can be generated in a few minutes on CPU, whereas labeling the test sets with IPOPT (for overall evaluation) takes several hours.
To assess generalization and robustness, we also construct extreme scenarios and out-of-distribution test scenarios.
Load demand and DER availability are sampled according to Table~\ref{tab:data_dist}, where load multipliers are relative to the OpenDSS base load, and DER availability is expressed as a fraction of rated capacity.
Switch statuses are drawn independently from a Bernoulli distribution with probability~$0.5$, so both meshed and radial topologies appear in training and testing.
Objective weights $(w_1, w_2)$ are sampled uniformly from $[0, 1]$ and $[0, 10]$, respectively.

\begin{table}[t]
\centering
\caption{Load and DER sampling ranges for training and test datasets.}
\label{tab:data_dist}
\scriptsize
\begin{tabular}{@{}lcc@{}}
\toprule
Dataset / scenario & Load multiplier & DER availability \\
\midrule
Training & $\mathcal{U}[0.2,\,1.5]$ & $\mathcal{U}[0,\,1]$ \\
\midrule
Base (in-distribution) & $\mathcal{U}[0.2,\,1.5]$ & $\mathcal{U}[0,\,1]$ \\
high-load/low-DER & $\mathcal{U}[0.6,\,1]$ & $\mathcal{U}[0.0,\,0.4]$ \\
low-load/high-DER & $\mathcal{U}[0.2,\,0.5]$ & $\mathcal{U}[0.6,\,1]$ \\
Distributional shift & $\mathcal{U}[0.2,\,1.5] \pm \cN(\mu, \sigma)$ & $\mathcal{U}[0,1] \pm \cN(\mu, \sigma)$ \\
\bottomrule
\end{tabular}
\vspace{-0.2in}
\end{table}

\textbf{Metrics.} Let $\cT$ denote the test set.
For each $t \in \cT$, let $\vx^{(t)}$ be the final predicted solution, $f(\vx^{(t)})$ the objective in \eqref{eq:obj}, and $f^{\star,(t)}$ the IPOPT reference objective.
We define the maximum-violation function $V_\text{max} (\vx) = \max \{V(\vx), R_\text{pf}(\vx)\}$, where $R_\text{pf}(\vx)$ represents the maximum residual of power flow equations \eqref{eq:vi_mat} and \eqref{eq:power_balance}.
The evaluation metrics are
\begin{align*}
    \text{Mean vio.} &= \frac{1}{|\cT|} \sum_{t \in \cT} V_\text{max} (\vx^{(t)}),  \\
	\text{Abs. cost gap} &= \frac{1}{|\cT|} \sum_{t \in \cT} \left| f(\vx^{(t)}) - f^{\star,(t)} \right|, \\
    \text{Optimality gap} &= \frac{1}{|\cT|} \sum_{t \in \cT} \frac{f(\vx^{(t)}) - f^{\star,(t)}}{f^{\star,(t)}} \times 100\%.
\end{align*}

\textbf{Methods.} We compare the following methods:
\begin{itemize}[leftmargin=18pt]
    \item \textbf{Direct Penalty}: Predicts DER powers and voltages
    $(\vp^{\text{der}}, \vq^{\text{der}}, \vv)$
    and penalizes equality and inequality violations in the loss, without an embedded power-flow solver.
    \item \textbf{Penalty}: The proposed method \emph{without} SLFS at inference.
    \item \textbf{Penalty+SLFS}: The proposed method.
	\item \textbf{SLFS-FSNet}: The FSNet method of \cite{nguyen2025fsnet}, upgraded with the SMW updates and using  SLFS during both training and inference (See in Appendix~\ref{appendix:slfs_fsnet}).
    \item \textbf{IPOPT}: An interior-point solver with the HSL MA97 linear solver.
    \item \textbf{Linearized OPF}: The OPF problem under a fixed-point linearized power-flow model (Sec.~\ref{sec:linearized model}), solved with Clarabel via CVXPY.
\end{itemize}

All experiments are done in JAX (Python) and Julia, and run on an NVIDIA H200 GPU and an AMD EPYC 9474F CPU.
Additional implementation details are given in Appendix~\ref{appendix:metrics_imp_details}.

\subsection{AC-OPF Performance Across Feeder Sizes}
\subsubsection{Solution quality and speedup}

\begin{figure}
	\centering
    \setlength{\abovecaptionskip}{-2pt}
	\includegraphics[width=0.99\linewidth]{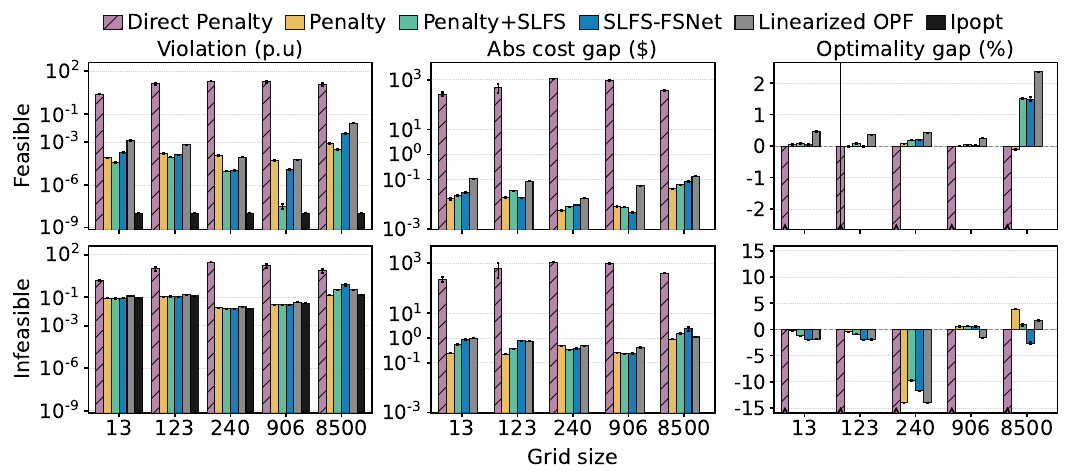}
	\caption{Test results for the base (in-distribution) test dataset.}
	\label{fig:main_results}
    \vspace{-0.1in}
\end{figure}

\begin{figure}
	\centering
    \setlength{\abovecaptionskip}{-1pt}
	\includegraphics[width=0.85\linewidth]{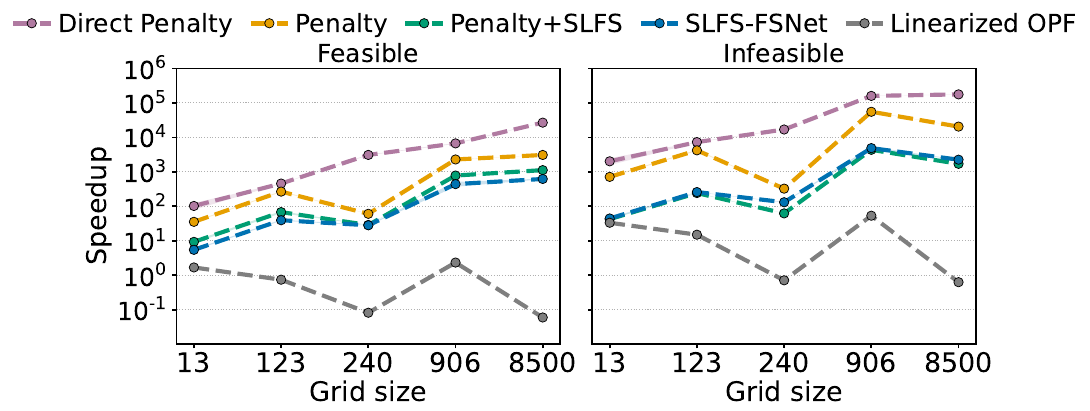}
	\caption{Speedup of the ML-based methods compared to IPOPT.}
	\label{fig:speedup}
    \vspace{-0.2in}
\end{figure}

Figure~\ref{fig:main_results} reports results on the base (in-distribution) test set.
Direct Penalty fails to converge on every feeder, whereas the remaining learning-based methods succeed.
Direct Penalty incurs severe constraint violations and high absolute cost gaps, and its negative relative optimality gaps are an artifact of severe infeasibility rather than improved solutions.  
This contrast indicates that embedding a power-flow solver in the training loop is essential for training success in the multiphase AC-OPF setting considered here.

On feasible instances, Penalty+SLFS attains the lowest constraint violations in all learning-based methods, achieving at most $3.0 \times 10^{-4}$ in the 8500-node feeder and as low as $3.2 \times 10^{-8}$ on the 906-bus feeder.
Its optimality gaps stay below $0.1\%$ on the 13-, 123-, 906-bus feeders, below $0.2\%$ on the 240-bus feeder, and reach $1.5\%$ only on the 8500-node feeder, where the absolute cost error is just \$0.06 but the objective values themselves are small, inflating the relative gap for all methods.
Relative to Penalty alone, enabling SLFS at inference reduces constraint violations by one to three orders of magnitude, at only a modest increase in absolute cost error.
SLFS-FSNet yields slightly higher violations than Penalty+SLFS, with comparable absolute cost errors and optimality gaps.
Noticeably, Penalty+SLFS outperforms Linearized OPF on all three solution-quality metrics on every feeder: its mean violations are at least $7\times$ and up to three orders of magnitude lower, and its absolute cost errors and optimality gaps are $1.6$--$7.4\times$ smaller.

On infeasible instances, Penalty+SLFS and SLFS-FSNet achieve constraint violations comparable to IPOPT.
On the smaller feeders (13-bus, 123-bus, and 240-bus), they can even report negative optimality gaps, meaning lower costs than the IPOPT reference.
When the AC-OPF problem is infeasible, however, the learning-based methods and IPOPT generally return different points that attain a similar violation level, so optimality gaps should be interpreted with care.

Figure~\ref{fig:speedup} summarizes solve speedups relative to IPOPT.
Direct Penalty is the fastest method but does not produce usable solutions.
Penalty is next, followed by Penalty+SLFS and SLFS-FSNet.
On feasible test sets, Penalty+SLFS is faster than SLFS-FSNet on every feeder---by $1.7$--$1.8\times$ except on the 240-bus feeder, where the two nearly coincide---because SLFS can be skipped when the network prediction is already feasible, whereas SLFS-FSNet always runs the repair procedure. 
On infeasible test sets, the two methods have similar speedups because SLFS is invoked on every sample.
Overall, Penalty+SLFS delivers one to three orders of magnitude speedup over IPOPT, reaching $780\times$ and $1121\times$ on the two largest feeders, while Linearized OPF is the slowest baseline and offers only limited acceleration.

In summary, Penalty+SLFS improves feasibility over Penalty, is typically faster than SLFS-FSNet, and delivers up to three orders of magnitude speedup over IPOPT.
It also outperforms Linearized OPF on constraint violation, optimality gap, and runtime, supporting learning-based surrogates as a practical alternative to linearized OPF for this problem class.

\subsubsection{Robustness to extreme regimes and distributional shift}

\begin{figure}[t]
	\centering
    \setlength{\abovecaptionskip}{-1pt}
	\includegraphics[width=0.99\linewidth]{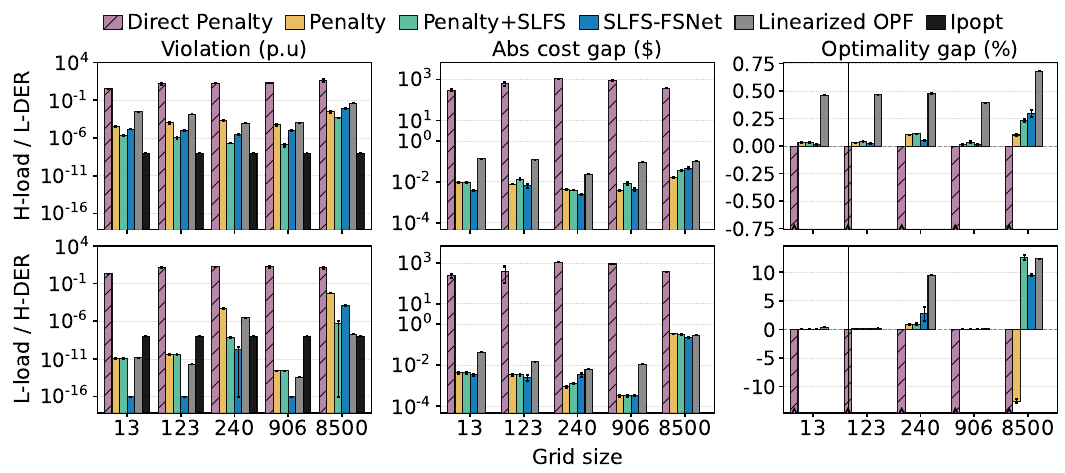}
	\caption{Test results on extreme scenarios.}
	\label{fig:extreme_scenarios}
    \vspace{-0.15in}
\end{figure}

\begin{figure}[t]
	\centering
    \setlength{\abovecaptionskip}{0pt}
	\includegraphics[width=0.75\linewidth]{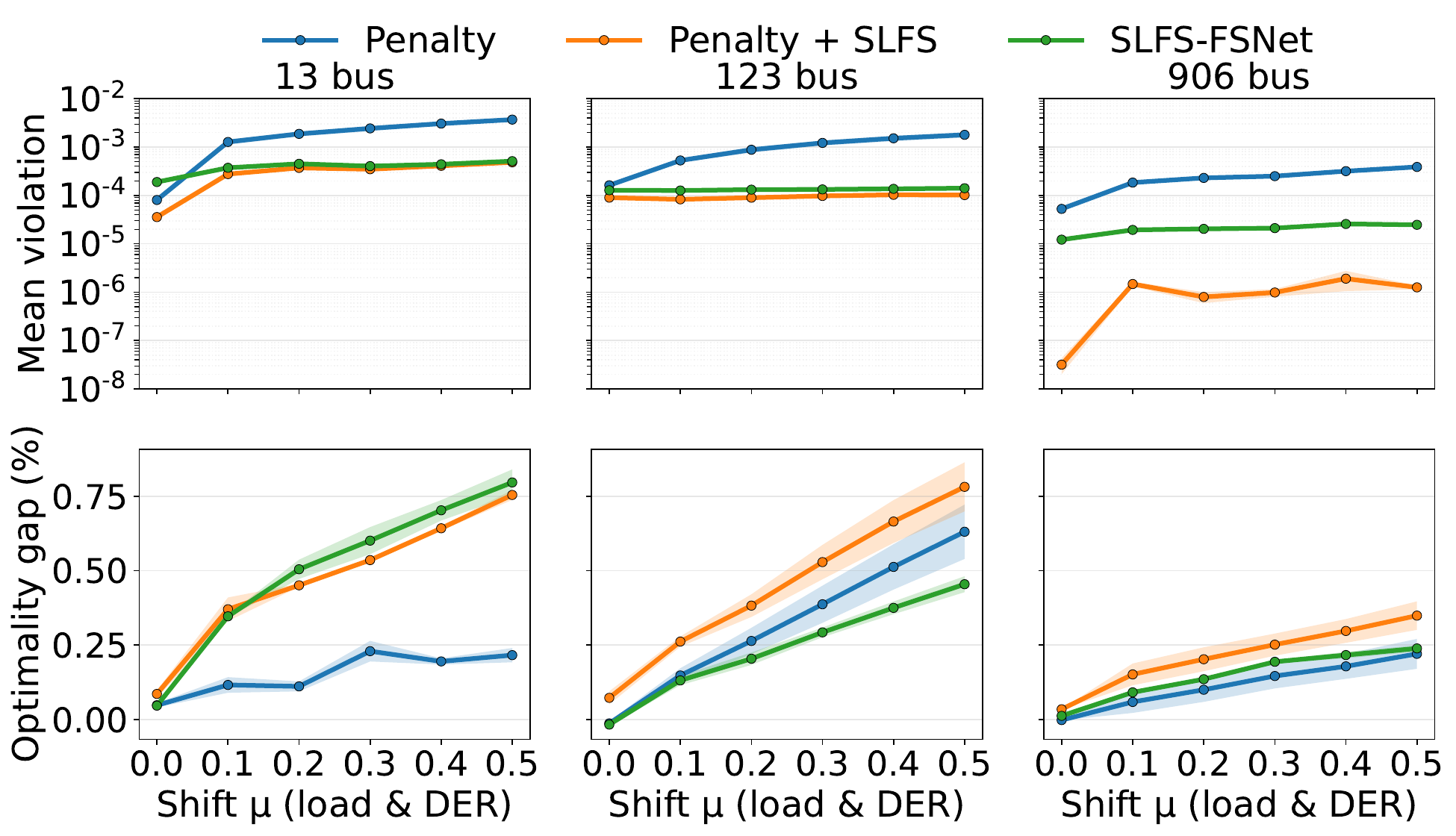}
	\caption{Robustness evaluation of the violation and optimality gap.}
	\label{fig:distributional_shift}
    \vspace{-0.15in}
\end{figure}

Figure~\ref{fig:extreme_scenarios} reports results on two extreme operating regimes: high-load/low-DER and low-load/high-DER.
In the high-load/low-DER regime, Penalty+SLFS again attains the lowest constraint violations, while its optimality gaps and absolute cost errors remain comparable to those of Penalty and SLFS-FSNet.
The largest optimality gap is $0.25\%$ on the 8500-node feeder and is substantially smaller on the other feeders.
Linearized OPF degrades relative to the base scenario and exhibits markedly higher absolute cost errors and optimality gaps than Penalty+SLFS.
The low-load/high-DER regime is easier for the learning-based methods: violations are low---even below IPOPT's on small and medium feeders---and optimality gaps remain tiny except on the 8500-node feeder.

Figure~\ref{fig:distributional_shift} evaluates robustness under distributional shift.
We set the Gaussian noise mean and standard deviation equal ($\mu = \sigma$) and vary $\mu$ from $0$ to $0.5$, generating 2{,}000 test samples per feeder at each noise level.
Constraint violations of Penalty increase steadily with noise, whereas Penalty+SLFS and SLFS-FSNet remain largely flat---a direct benefit of the SLFS repair step.
Optimality gaps rise only mildly for all three methods and stay below $1\%$ at all perturbation levels.
These results show that Penalty+SLFS continues to produce high-quality solutions under substantial distributional shifts in load demand and PV availability.
This robustness stems in part from its self-supervised training procedure, which makes training data inexpensive to generate and allows the model to be trained on large datasets spanning diverse operating regimes.

\subsection{Ablations and Computational Analysis}
\subsubsection{Differentiating the power flow solver}

\begin{table}[t]
    \caption{Differentiation of the power flow solver.}
\label{tab:differentiation_of_power_flow_solver}
    \centering
    \setlength{\tabcolsep}{2pt}
    \fontsize{6}{6.5}\selectfont
    \begin{tabular}{l l c c c c c}
        \toprule
        \multirow{2}{*}{Grid} & \multirow{2}{*}{Method} & Mean & Max & Peak GPU & Time & \multirow{2}{*}{Speedup} \\
        & & error & error & (GB) & /batch (ms) & \\
        \midrule
        \multirow{6}{*}{13 bus} & Full & 0 & 0 & 0.06 & 0.24 & 1.00 \\
         & 1-Step & 0.01 & 0.13 & 0.06 & 0.15 & 1.66 \\
         & 2-Step & 2.27e-03 & 0.02 & 0.06 & 0.17 & 1.43 \\
         & 3-Step & 2.72e-04 & 3.71e-03 & 0.06 & 0.18 & 1.33 \\
         & 4-Step & 5.59e-05 & 7.57e-04 & 0.06 & 0.19 & 1.30 \\
         & Implicit & 4.98e-11 & 8.25e-10 & 0.07 & 2.11 & 0.11 \\
        \midrule
        \multirow{6}{*}{123 bus} & Full & 0 & 0 & 0.10 & 1.11 & 1.00 \\
         & 1-Step & 8.02e-03 & 0.10 & 0.10 & 0.61 & 1.82 \\
         & 2-Step & 1.16e-03 & 0.01 & 0.10 & 0.77 & 1.43 \\
         & 3-Step & 9.23e-05 & 1.43e-03 & 0.10 & 0.83 & 1.33 \\
         & 4-Step & 1.37e-05 & 1.82e-04 & 0.10 & 0.79 & 1.40 \\
         & Implicit & 5.78e-13 & 2.33e-10 & 0.43 & 13.41 & 0.08 \\
        \midrule
        \multirow{6}{*}{906 bus} & Full & 0 & 0 & 3.98 & 298.72 & 1.00 \\
         & 1-Step & 0.01 & 0.13 & 3.98 & 67.39 & 4.43 \\
         & 2-Step & 1.18e-03 & 0.01 & 4.80 & 88.94 & 3.36 \\
         & 3-Step & 7.48e-05 & 8.16e-04 & 4.80 & 108.97 & 2.74 \\
         & 4-Step & 7.99e-06 & 8.07e-05 & 4.80 & 130.16 & 2.29 \\
         & Implicit & 3.78e-14 & 7.35e-13 & 41.85 & 1388.29 & 0.22 \\
        \bottomrule
    \end{tabular}
    \vspace{-0.2in}
\end{table}

Table~\ref{tab:differentiation_of_power_flow_solver} compares the $M$-step Jacobian approximation ($M=1,2,3,4$) with implicit differentiation for the power-flow solver.
Implicit differentiation is the most accurate option, but it is also the slowest---even slower than full unrolled differentiation---and has the largest memory footprint.
For the $M$-step approximation, both accuracy and runtime increase with $M$.
The \mbox{2-,} \mbox{3-,} and 4-step variants use similar memory, likely due to JAX kernel fusion.
These approximations also yield substantial speedups over full unrolled differentiation. For example, on the 906-bus feeder, the 3-step approximation is $2.74\times$ faster while remaining accurate enough for training.
We therefore adopt the 3-step approximation in all experiments as a practical balance between accuracy and cost.
The end-to-end results in Figures~\ref{fig:main_results} and \ref{fig:extreme_scenarios} confirm that this choice is sufficient to train high-quality models.
More broadly, the results suggest that approximate Jacobian information is adequate for training, and that unrolled differentiation through the power-flow solver can be made efficient because the fixed-point iteration converges fast and each step is a simple matrix--vector product.

\subsubsection{SMW update for topology change}\label{sec:smw_results}

\begin{table}[t]
    \caption{Accuracy and speedup of the SMW update for matrix inversion.}
    \label{tab:speedup_of_smw_update}
    \centering
    \setlength{\tabcolsep}{2pt}
    \fontsize{6}{6.5}\selectfont
    \begin{tabular}{l c r r r r r}
        \toprule
        Grid & IR Steps & Inverse Diff. & Max Residual & SMW (ms) & Direct (ms) & Speedup \\
        \midrule
        \multirow{3}{*}{13 bus} & 0 & 3.71e-05 & 1.20e-04 & 1.12 & \multirow{3}{*}{1.24} & 1.11 \\
         & 1 & 4.77e-09 & 1.75e-08 & 1.51 &  & 0.82 \\
         & 2 & 1.23e-10 & 2.17e-10 & 1.89 &  & 0.66 \\
        \midrule
        \multirow{3}{*}{123 bus} & 0 & 4.90e-04 & 3.59e-04 & 1.94 & \multirow{3}{*}{7.19} & 3.71 \\
         & 1 & 1.56e-05 & 1.74e-07 & 2.63 &  & 2.73 \\
         & 2 & 1.53e-05 & 5.33e-10 & 3.29 &  & 2.19 \\
        \midrule
        \multirow{3}{*}{240 bus} & 0 & 1.89e-03 & 2.05e-03 & 2.41 & \multirow{3}{*}{29.08} & 12.09 \\
         & 1 & 1.33e-05 & 1.18e-06 & 5.71 &  & 5.10 \\
         & 2 & 9.22e-06 & 2.20e-09 & 8.69 &  & 3.35 \\
        \midrule
        \multirow{3}{*}{906 bus} & 0 & 2.09e-11 & 1.50e-10 & 7.30 & \multirow{3}{*}{140.97} & 19.31 \\
         & 1 & 2.64e-11 & 1.97e-11 & 57.26 &  & 2.46 \\
         & 2 & 2.71e-11 & 1.78e-11 & 110.25 &  & 1.28 \\
        \midrule
        \multirow{3}{*}{8500 node} & 0 & 8.86e-08 & 3.47e-09 & 58.90 & \multirow{3}{*}{1666.63} & 28.30 \\
         & 1 & 1.67e-09 & 7.41e-11 & 1541.56 &  & 1.08 \\
         & 2 & 1.66e-09 & 6.86e-11 & 3023.57 &  & 0.55 \\
        \bottomrule
    \end{tabular}
    \vspace{-0.2in}
\end{table}

Table~\ref{tab:speedup_of_smw_update} compares the SMW update of the inverse admittance matrix with direct inversion at batch size~8.
Without iterative refinement (IR), SMW speedups grow steadily with feeder size---from $1.1\times$ on the 13-bus system to $12\times$, $19\times$, and $28\times$ on the 240-bus, 906-bus, and 8500-node feeders---while direct inversion slows from $1.2$\,ms to $1.7$\,s.
IR is useful mainly on the smaller feeders: without it, residuals on the 13-, 123-, and 240-bus systems are on the order of $10^{-4}$--$10^{-3}$, and a single IR step cuts them by several orders of magnitude (e.g., to ${\sim} 10^{-6}$ on 240-bus) at only a moderate runtime cost, with a second step giving diminishing returns.
On the 906-bus and 8500-node feeders, by contrast, the residual is already $\lesssim 10^{-9}$ with no IR, so refinement is unnecessary for these feeders.
In practice, we therefore use SMW alone on large feeders and add 1-2 SMW refinement iterations only on smaller systems.

\subsubsection{SLFS vs.\ least-squares solver}

\begin{table}[t]
    \caption{Comparison of SLFS and the least-squares solver for feasibility seeking.}
    \label{tab:comparison_of_slfs_and_least_squares_solver}
    \centering
    \setlength{\tabcolsep}{2pt}
    \fontsize{6}{6.5}\selectfont
    \begin{tabular}{lccccccc}
        \toprule
        Grid & SLFS (ms) & LM (ms) & Speedup & SLFS eq & SLFS ineq & LM eq & LM ineq \\
        \midrule
        13 bus & 0.66 & 59.48 & 89.49 & 1.49e-05 & 1.18e-05 & 3.23e-04 & 1.45e-04 \\
        123 bus & 1.60 & 161.40 & 100.85 & 6.79e-07 & 2.30e-05 & 7.83e-05 & 7.55e-05 \\
        240 bus & 10.63 & 12900.49 & 1.21e+03 & 5.17e-06 & 1.59e-06 & 0.01 & 4.52e-11 \\
        \bottomrule
    \end{tabular}
    \vspace{-0.2in}
\end{table}

Table~\ref{tab:comparison_of_slfs_and_least_squares_solver} compares SLFS with the least-squares feasibility seeker from FSNet~\cite{nguyen2025fsnet}, solved by Levenberg--Marquardt (LM) method.
SLFS is $89\times$--$1210\times$ faster across the three feeders.
It also yields lower equality residuals in every case and lower inequality residuals on the 13- and 123-bus feeders; on the 240-bus feeder, LM attains a smaller inequality residual but leaves a much larger equality residual ($ {\sim} 10^{-2}$ vs.\ $10^{-6}$).
Thus SLFS is both more accurate overall and far more scalable as feeder size grows.

\subsubsection{Training time and memory}

\begin{table}[t]
	\centering
	\caption{Training time and memory usage of the Penalty+SLFS method.}
	\label{tab:training_time_and_memory_usage}
	\fontsize{5.5}{6.5}\selectfont
	\setlength{\tabcolsep}{1pt}
	\begin{tabular}{lccccccccc}
		\toprule
		& \multicolumn{3}{c}{Direct Penalty} & \multicolumn{3}{c}{Penalty+SLFS} & \multicolumn{3}{c}{SLFS-FSNet} \\
		\cmidrule(lr){2-4} \cmidrule(lr){5-7} \cmidrule(lr){8-10}
		Grid & Epochs & Time (h) & Peak Mem (GB) & Epochs & Time (h) & Peak Mem (GB) & Epochs & Time (h) & Peak Mem (GB) \\
		\midrule
		13 bus & 20 & 0.15\std{0.05} & 0.52\std{0.00} & 10 & 0.09\std{0.01} & 0.37\std{0.00} & 10 & 0.09\std{0.01} & 0.68\std{0.00} \\
		123 bus & 20 & 0.47\std{0.24} & 2.30\std{0.00} & 5 & 0.05\std{0.01} & 2.30\std{0.00} & 5 & 0.23\std{0.04} & 2.30\std{0.00} \\
		240 bus & 20 & 0.22\std{0.00} & 8.19\std{0.00} & 5 & 0.31\std{0.00} & 8.19\std{0.00} & 5 & 0.83\std{0.00} & 8.19\std{0.00} \\
		906 bus & 20 & 0.78\std{0.00} & 17.08\std{0.00} & 5 & 0.25\std{0.00} & 17.08\std{0.00} & 5 & 0.62\std{0.03} & 17.08\std{0.00} \\
		8500 node & 10 & 3.72\std{0.05} & 68.60\std{0.00} & 5 & 2.29\std{0.02} & 68.60\std{0.00} & 5 & 6.75\std{0.02} & 68.60\std{0.00} \\
		\bottomrule
	\end{tabular}
    \vspace{-0.1in}
\end{table}

Table~\ref{tab:training_time_and_memory_usage} compares training time and peak GPU memory across learning-based methods.
Direct Penalty generally requires more wall-clock time than Penalty+SLFS because it uses more epochs.
At a matched epoch budget, Penalty+SLFS has similar training time to SLFS-FSNet on the 13-bus feeder, but is $2.46$--$4.6\times$ faster on the larger feeders.
From the 123-bus feeder onward, peak GPU memory is nearly identical across methods, as it is dominated by the admittance matrix and its inverse.

\subsubsection{Training regimes}

\begin{figure}[t]
	\centering
    \setlength{\abovecaptionskip}{0pt}
	\includegraphics[width=0.78\linewidth]{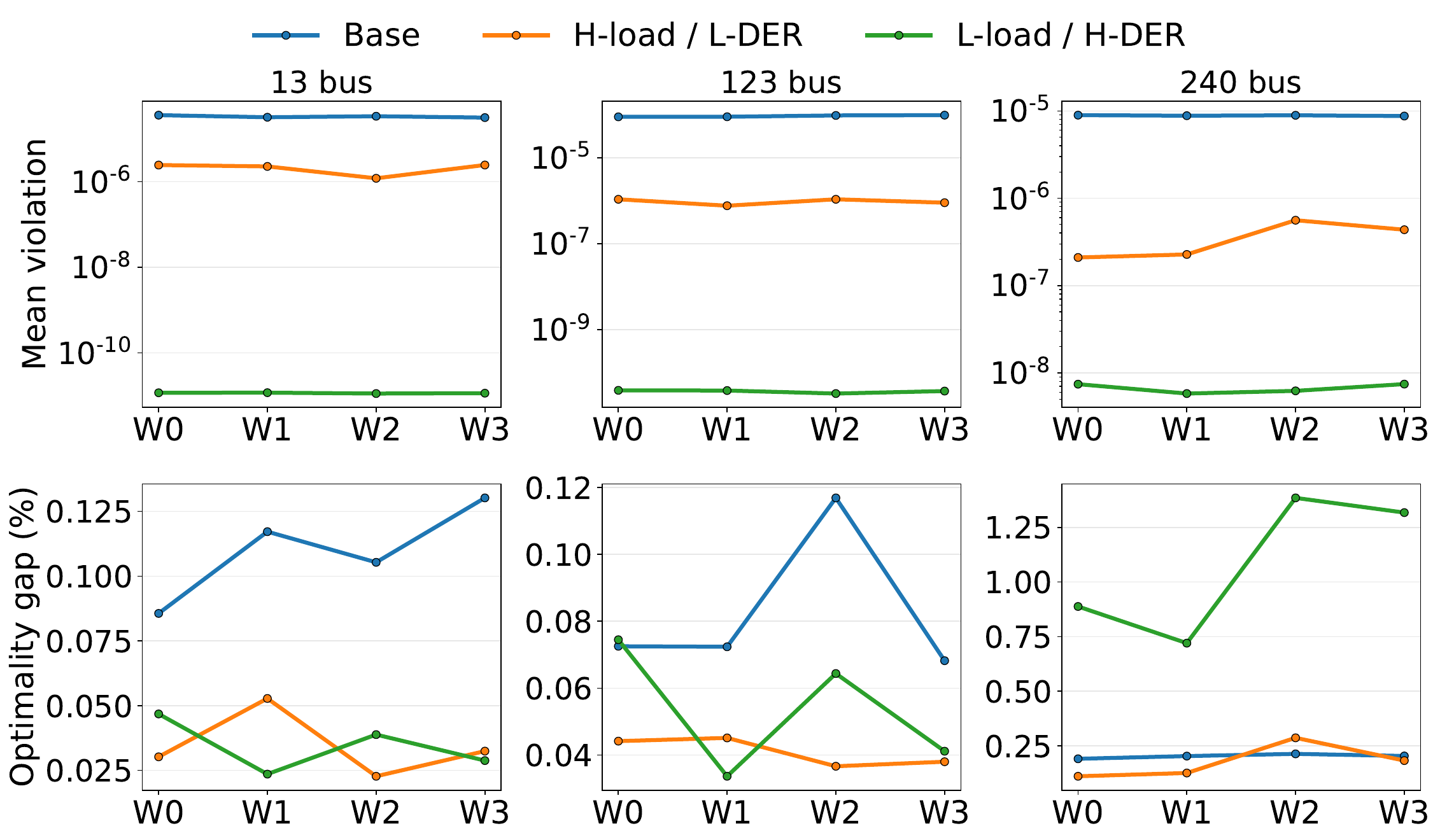}
	\caption{The Penalty+SLFS method with different training datasets.}
	\label{fig:regime_ablation}
    \vspace{-0.15in}
\end{figure}

To study how the training distribution affects Penalty+SLFS, we vary the mix of load/DER regimes while keeping the learning algorithm and hyperparameters fixed.
Sampling weights for regimes W0--W3 are  in Table~\ref{tab:regime-defs-weights}.

\begin{table}[t]
	\centering
	\caption{Load/DER training regimes and sampling weights W0--W3.}
	\label{tab:regime-defs-weights}
    \fontsize{6}{6.5}\selectfont
	\begin{tabular}{@{}lcc|@{\hspace{0.4em}}cccc@{}}
	  \toprule
	  Regime & Load scale & DER avail. & W0 & W1 & W2 & W3 \\
	  \midrule
	  Base                 & [0.2, 1.5] & [0.0, 1.0] & 1.00 & 0.60 & 0.60 & 0.60 \\
	  High load, low DER  & [0.7, 1.0] & [0.0, 0.3] & 0.00 & 0.05 & 0.35 & 0.20 \\
	  Low load, high DER  & [0.2, 0.5] & [0.6, 1.0] & 0.00 & 0.35 & 0.05 & 0.20 \\
	  \bottomrule
	\end{tabular}
    \vspace{-0.1in}
\end{table}

Figure~\ref{fig:regime_ablation} reports mean violation and optimality gap for Penalty+SLFS under the four training mixes W0--W3, evaluated on the base, high-load/low-DER, and low-load/high-DER test regimes.
Across all three feeders, mean violations remain small for every mix: typically $10^{-4}$--$10^{-5}$ on the base test set, lower under high-load/low-DER, and lowest under low-load/high-DER (often by several orders of magnitude).
The curves are nearly flat in W0--W3, so the training mix has little effect on feasibility once SLFS is applied at inference.
Optimality gaps, by contrast, depend more on both the test regime and the feeder.
On the 13-bus and 123-bus systems, gaps stay below roughly $0.15\%$ for all regimes, with the base test set usually the hardest and the extreme regimes somewhat easier; the high-load/low-DER gap is especially stable across mixes.
On the 240-bus feeder, the low-load/high-DER regime is the outlier: its gap rises to about $1.4\%$ at W2, whereas the base and high-load/low-DER gaps remain in the $0.1$--$0.3\%$ range.
Overall, there is no consistent ranking among W0--W3, and overweighting a single extreme regime (as in W1 or W2) does not reliably improve that regime's optimality.
This suggests that the broad base distribution already covers much of the operating space, while SLFS keeps constraint satisfaction robust regardless of how the training data is mixed.

\begin{figure}[t]
	\centering
    \setlength{\abovecaptionskip}{-1pt}
	\includegraphics[width=0.75\linewidth]{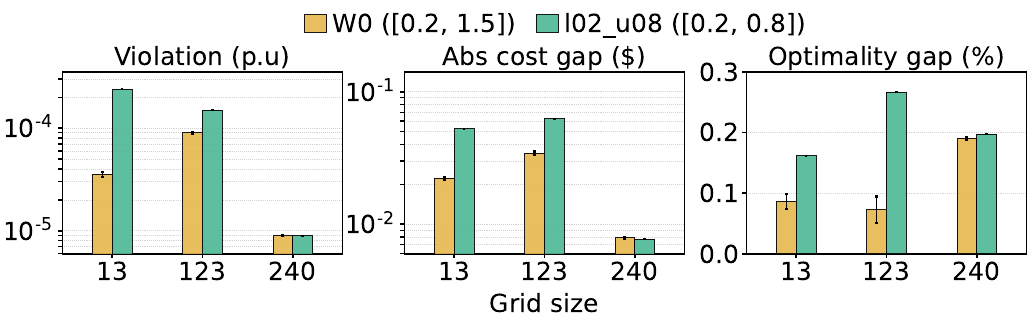}
	\caption{Performance of the Penalty+SLFS method with training datasets with load range $[0.2, 0.8]$ and tested on the load range $[0.2, 1.5]$.}
	\label{fig:l02_u08_ablation_feasible}
    \vspace{-0.2in}
\end{figure}

We further train on loads restricted to $[0.2, 0.8]$ and evaluate on the wider range $[0.2, 1.5]$ (Figure~\ref{fig:l02_u08_ablation_feasible}).
Constraint violations, cost gaps, and optimality gaps rise only slightly relative to the base setting. 
Violations remain on the order of $10^{-3}$ and optimality gaps stay below $0.3\%$, which is still acceptable for practical use.
As training-data generation is cheap in our self-supervised approach, we recommend training over the broadest operating range to maximize solution quality and robustness.

Additional experimental results on the penalty weight $\rho$, the inclusion of switch statuses in the neural-network input, MLP versus graph neural network backbones, and inference-time scaling with batch size are reported in Appendix~\ref{appendix:exp_results}.

\section{Conclusion}
This paper introduced Penalty+SLFS, a scalable self-supervised framework for multiphase distribution AC-OPF under switch-induced topology reconfiguration. The framework trains a neural network directly by minimizing the AC-OPF objective and constraint penalties, using a clipping layer and a differentiable fixed-point power-flow solver, without requiring labeled optimal solutions. 
To enable efficient training, an $M$-step Jacobian approximation--with approximation error that decays geometrically in $M$--reduces the time and memory costs of differentiation relative to implicit differentiation. Sherman--Morrison--Woodbury updates efficiently account for switch-induced changes in the inverse admittance matrix. At inference time, sequential linearized feasibility seeking (SLFS) repairs constraint violations whenever they exceed a prescribed tolerance; SLFS provably reduces violations monotonically and relies primarily on GPU-friendly matrix--vector operations.
Numerical results on IEEE feeders from 13 to 8{,}500 nodes show that Penalty+SLFS produces high-quality solutions with near-zero constraint violations, optimality gaps below $0.2\%$ on small-medium feeders and $1.5\%$ on the 8500-node feeder, and speedups of $780\times$ and $1121\times$ over IPOPT on the two largest feeders.
Leveraging large and diverse training sets, the framework remains robust under extreme load/DER regimes and distributional shifts, with optimality gaps staying below $1\%$ at the largest perturbation level.
Together, this shows the potential of surrogate learning methods to learn fast, feasible, and robust approximations to multiphase distribution AC-OPF.

\bibliographystyle{IEEEtran}
\bibliography{references}

\clearpage
\appendix 
\renewcommand\thefigure{\thesection.\arabic{figure}}    
\setcounter{figure}{0}    
\renewcommand\thetable{\thesection.\arabic{table}}    
\setcounter{table}{0}

\subsection{Additional experimental results} \label{appendix:exp_results}
Figure~\ref{fig:rho_ablation_feasible} studies the effect of the penalty weight $\rho$, used in \eqref{eq:training_loss}.
Constraint violations decrease as $\rho$ increases, as expected when violations are penalized more heavily.
Absolute cost and optimality gaps are smallest at $\rho = 10^{-5}$ and are typically largest at the most aggressive setting $\rho = 10^{6}$.

Figure~\ref{fig:topo_nn_input_ablation} ablates whether switch statuses are provided as neural-network inputs.
Models that receive switch statuses achieve substantially smaller absolute cost and optimality gaps, underscoring the value of topology information in the input.

Figure~\ref{fig:mlp_vs_gcn} compares the test results of the MLP and GCN backbones in the Penalty + SLFS method, which are trained with the same training pipeline and budget.
While both backbones achieve similar feasibility, the MLP backbone achieves consistently lower absolute cost gaps and optimality gaps than the GCN backbone in all cases. 
That the difference appears only in optimality and not in feasibility is consistent with the design of the framework, in which SLFS enforces the constraints regardless of the backbone.

In addition, Figure~\ref{fig:batch_size_inference_time} shows how full-pipeline inference wall time scales with batch size. On the 13-, 123-, and 240-bus systems, batch time grows sublinearly ($t \propto B^{\alpha}$ with $\alpha = 0.38$, $0.51$, $0.55$):
increasing $B$ from 1 to 64 raises wall time by only $4.5$--$9.7\times$, yielding $6.6$--$14\times$ higher per-sample throughput and up to $1.6 \times 10^{4}$ solutions/s on the 13-bus feeder. On the 906-bus and 8500-node systems the trend becomes linear to slightly superlinear ($\alpha = 0.97$ and $1.27$), so throughput peaks at $B = 8$ and $B = 1$ respectively and then degrades, and large batches hit GPU memory limits (8500-node gets an out-of-memory error at $B = 64$).

\subsection{Hyperparameters and implementation details}\label{appendix:metrics_imp_details}

Training hyperparameters are summarized in Table~\ref{tab:hyperparameters}.
The $M$-step Jacobian truncation uses $M=3$.
SLFS uses $K^{\mathrm{lin}}=2$ linearization steps and proximal weight $\eta=0$.
For Penalty+SLFS, the penalty weight $\rho$ is set to $10^{5}$. 

Power flow is solved with the fixed-point iteration of Section~\ref{sec:power_flow_solver} on the 13-bus, 123-bus, 906-bus, and 8500-node feeders. On the 240-bus feeder we use Newton--Raphson, because the fixed-point iteration does not converge.

\subsection{SLFS-FSNet}\label{appendix:slfs_fsnet}

The architecture of SLFS-FSNet is adapted from FSNet~\cite{nguyen2025fsnet}, which uses a feasibility-seeking algorithm to correct the neural-network prediction into a feasible solution.
The SLFS module is embedded in both training and inference to enforce constraint satisfaction: it takes raw DER setpoints $\vp^{\text{der,raw}}, \vq^{\text{der,raw}}$ as input and outputs feasible setpoints $\vp^{\text{der}}, \vq^{\text{der}}$.

The original FSNet~\cite{nguyen2025fsnet} assumes a fixed topology and solves the feasibility-seeking problem with the Levenberg--Marquardt method.
We instead embed our SMW update and SLFS module in FSNet to handle topology changes and accelerate feasibility seeking (Figure~\ref{fig:fsnet_framework}).


\begin{table}[t]
    \centering
    \caption{Hyperparameters}
    \scriptsize
    \setlength{\tabcolsep}{3pt}
    \begin{tabular}{ll|ll}
    \toprule
    \textbf{Parameter} & \textbf{Value} & \textbf{Parameter} & \textbf{Value} \\
    \midrule
    Optimizer & Adam & LR scheduler & Cosine \\
    Learning rate & $5\times10^{-5}$ & Batch size & 4-32 \\
    Hidden layers & 4 & Nonlinearity & Leaky ReLU \\
    Jacobian steps ($M$) & 3 & Lin.\ steps ($K^{\mathrm{lin}}$) & 2 \\
    Proximal weight $\eta$ & $0$ & Random seed & $1,2,3$ \\
    Penalty $\rho$ (Penalty) & $10^{5}$ & Distance $\gamma$ (FSNet) & $1$ \\
    \bottomrule
    \end{tabular}
    \label{tab:hyperparameters}
\end{table}

\begin{figure}[t]
	\centering
	\includegraphics[width=0.9\linewidth]{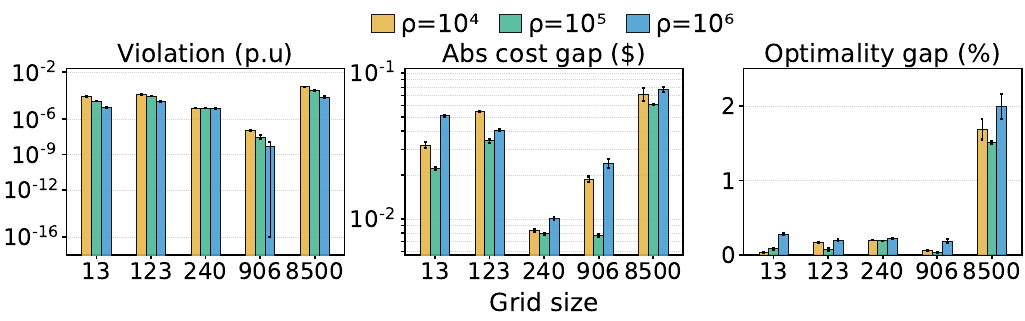}
	\caption{Performance of the Penalty + SLFS method with different penalty weights $\rho$.}
	\label{fig:rho_ablation_feasible}
\end{figure}

\begin{figure}[t]
	\centering
	\includegraphics[width=0.9\linewidth]{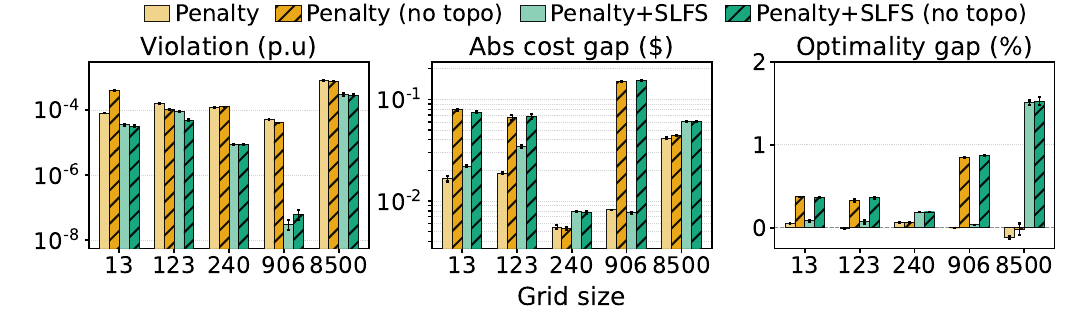}
	\caption{Performance of the models with and without the switch statuses in the NN input.}
	\label{fig:topo_nn_input_ablation}
\end{figure}

\begin{figure}[t]
	\centering
	\setlength{\subfigcapskip}{-4pt}
	\setlength{\abovecaptionskip}{0pt}
	\subfigure[Base scenario]{%
		\includegraphics[width=0.8\linewidth]{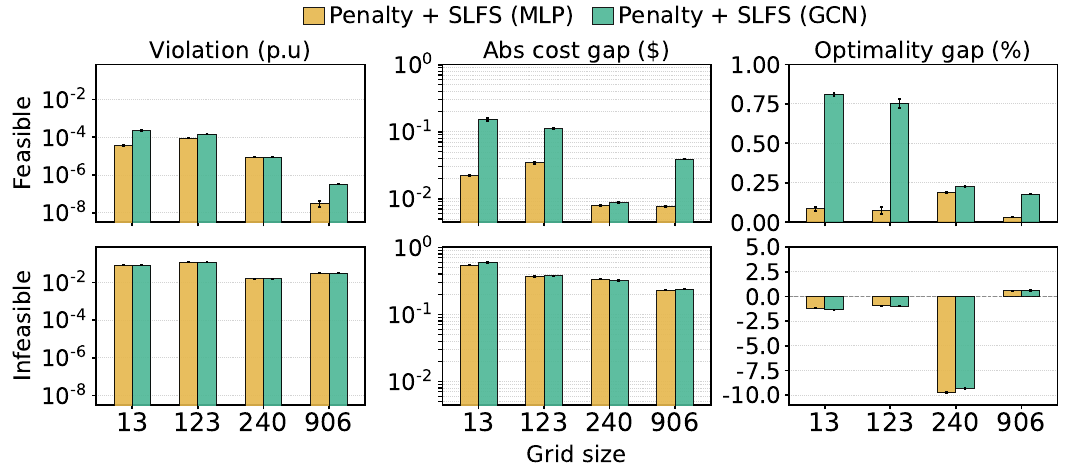}
		\label{fig:mlp_vs_gcn_base}
	}\\[-0.4em]
	\subfigure[Extreme scenarios]{%
		\includegraphics[width=0.8\linewidth]{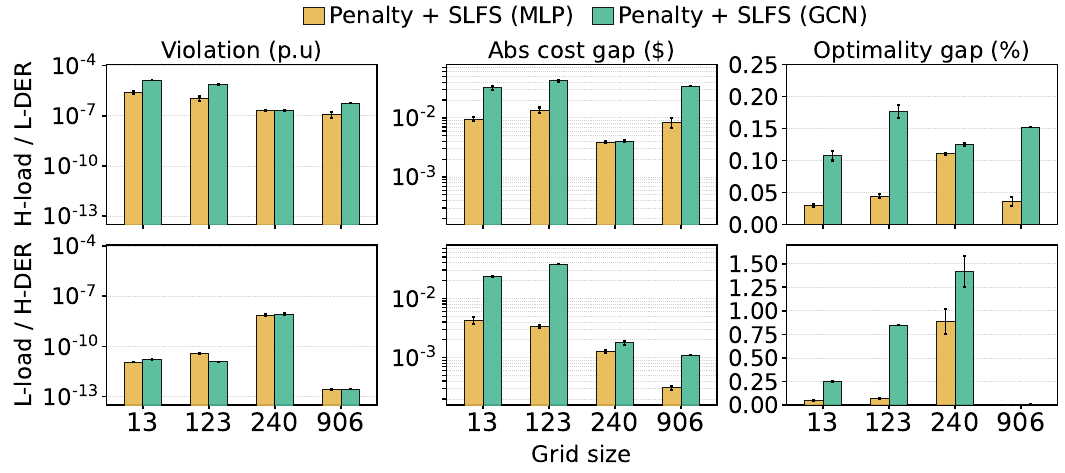}
		\label{fig:mlp_vs_gcn_regime}
	}
	\caption{MLP vs.\ GCN test results.}
	\label{fig:mlp_vs_gcn}
\end{figure}

\subsection{Linearized network constraints}\label{appendix:lin_constraints}

This appendix continues the linearized constraint model in Sec.~\ref{sec:slfs}, building on the affine voltage model \eqref{eq:lin_pf_simple} and the voltage-magnitude linearization \eqref{eq:vmag_matrix}--\eqref{eq:Av_vmag}.

\textbf{Substation power.}
The substation limit \eqref{eq:s0_lim} involves $|\vs_0|$, where \eqref{eq:vs0_mat} defines the per-phase slack-bus injection
$\vs_0 = \diag(\vv_0)\bigl(\overline{\vY}_{00}\overline{\vv}_0 + \overline{\vY}_{0L}(\vz)\overline{\vv}\bigr)$.
Because the slack voltage $\vv_0$ is fixed, $\vs_0$ depends on $\vx$ only through the load-bus voltages $\vv$.
At the reference point,
$\hvs_0 \coloneqq \diag(\vv_0)\bigl(\overline{\vY}_{00}\overline{\vv}_0 + \overline{\vY}_{0L}(\vz)\overline{\hvv}\bigr)$,
and substituting \eqref{eq:lin_pf_simple} gives
\begin{align*}
	\vs_0 &\approx \hvs_0 + \diag(\vv_0)\,\overline{\vY_{0L}(\vz)\,\vB}\,(\vx-\hvx).
\end{align*}
Applying the same magnitude linearization as in \eqref{eq:vmag_matrix}--\eqref{eq:Av_vmag} yields
\begin{align}
	|\vs_0| &\approx |\hvs_0| + \vA_s(\hvx) (\vx - \hvx), \nonumber \\
	\vA_s(\hvx) &\coloneqq \diag(|\hvs_0|)^{-1} \Re\left\{\diag(\overline{\hvs_0})\,\diag(\vv_0)\,\overline{\vY_{0L}(\vz)\,\vB} \right\},
	\label{eq:As_sub}
\end{align}
where $|\vs_0| \coloneqq (|s_{0,\phi}|)_{\phi \in \Phi_0}$.

\textbf{Angle-difference.}
For each static branch $(m,n)\in\cE$, let $v_m$ and $v_n$ denote the terminal phase voltages.
The angle-difference limit \eqref{eq:angle_lim} can be rewritten with
$\xi_{mn} \coloneqq \overline{v_m}\,v_n$ and $\kappa_{mn} \coloneqq \tan(\alpha_{mn}^{\mathrm{max}})$ as the one-sided pair
\begin{align}
	&\Im\{\xi_{mn}\} - \kappa_{mn}\,\Re\{\xi_{mn}\} \leq 0, \nonumber \\
	&-\Im\{\xi_{mn}\} - \kappa_{mn}\,\Re\{\xi_{mn}\} \leq 0.
	\label{eq:angle_pair}
\end{align}
To linearize $\xi_{mn}$, extend \eqref{eq:lin_pf_simple} to all phase nodes by stacking
$\vv^{\mathrm{full}} = (\vv^\top,\vv_0^\top)^\top$ and writing
$v_n \approx \hat{v}_n + \vB^{\mathrm{full}}_{n,:}(\vx-\hvx)$,
where $\hat{v}_n$ is the $n$-th entry of $\vv^{\mathrm{full}}$ at $(\hvx,\hvv)$ and $\vB^{\mathrm{full}}$ equals $\vB$ on rows $n\in\cP$ and is zero on the fixed slack rows.
For branch $(m,n)$, linearizing $\xi_{mn}=\overline{v_m}\,v_n$ around $(\hvx,\hvv)$ gives
\begin{align*}
	\hat{\xi}_{mn} &\coloneqq \overline{\hat{v}_m}\,\hat{v}_n, \\
	\vD_{mn} &\coloneqq \overline{\vB^{\mathrm{full}}_{m,:}}\,\hat{v}_n + \overline{\hat{v}_m}\,\vB^{\mathrm{full}}_{n,:} \in \C^{1\times 2N_{\mathrm{c}}},
\end{align*}
so $\xi_{mn} \approx \hat{\xi}_{mn} + \vD_{mn}(\vx-\hvx)$.
Let $\hat{\boldsymbol{\xi}} \in \C^{|\cE|}$ and $\vD \in \C^{|\cE|\times 2N_{\mathrm{c}}}$ stack $(\hat{\xi}_{mn})_{(m,n)\in\cE}$ and $(\vD_{mn})_{(m,n)\in\cE}$, and let $\boldsymbol{\kappa} \coloneqq (\kappa_{mn})_{(m,n)\in\cE}$.
Substituting these affine expressions into \eqref{eq:angle_pair} yields
\begin{align}
	\vb_{\mathrm{ang}}(\vx) &\approx \vb_{\mathrm{ang}}(\hvx) + \vA_{\mathrm{ang}}(\hvx)(\vx - \hvx), \nonumber \\
	\vb_{\mathrm{ang}}(\hvx) &\coloneqq
	\begin{bmatrix}
		\Im\{\hat{\boldsymbol{\xi}}\} - \diag(\boldsymbol{\kappa})\,\Re\{\hat{\boldsymbol{\xi}}\} \\
		-\Im\{\hat{\boldsymbol{\xi}}\} - \diag(\boldsymbol{\kappa})\,\Re\{\hat{\boldsymbol{\xi}}\}
	\end{bmatrix}, \label{eq:g_ang} \\
	\vA_{\mathrm{ang}}(\hvx) &\coloneqq
	\begin{bmatrix}
		\Im\{\vD\} - \diag(\boldsymbol{\kappa})\,\Re\{\vD\} \\
		-\Im\{\vD\} - \diag(\boldsymbol{\kappa})\,\Re\{\vD\}
	\end{bmatrix}, \label{eq:A_ang}
\end{align}
where $\vb_{\mathrm{ang}}(\vx)$ stacks the left-hand sides of \eqref{eq:angle_pair} over $(m,n)\in\cE$.

\textbf{Branch current.}
Branch-current limits \eqref{eq:i_lim} depend on $|i_{mn}|$, where \eqref{eq:i_def} gives
$i_{mn} = y_{mn}\,(v_m - v_n)$.
Using the extended affine voltage model above,
\begin{align*}
	\hat{i}_{mn} &\coloneqq y_{mn}\,(\hat{v}_m - \hat{v}_n), \\
	\vD^{\mathrm{br}}_{mn} &\coloneqq y_{mn}\,(\vB^{\mathrm{full}}_{m,:} - \vB^{\mathrm{full}}_{n,:}) \in \C^{1\times 2N_{\mathrm{c}}},
\end{align*}
so $i_{mn} \approx \hat{i}_{mn} + \vD^{\mathrm{br}}_{mn}(\vx-\hvx)$.
Let $\hat{\vi}^{\mathrm{br}} \in \C^{|\cE|}$ and $\vD^{\mathrm{br}} \in \C^{|\cE|\times 2N_{\mathrm{c}}}$ stack $(\hat{i}_{mn})_{(m,n)\in\cE}$ and $(\vD^{\mathrm{br}}_{mn})_{(m,n)\in\cE}$.
The same magnitude linearization as in \eqref{eq:vmag_matrix}--\eqref{eq:Av_vmag} gives, for each branch,
$|i_{mn}| \approx |\hat{i}_{mn}| + \vA^{\mathrm{br}}_{mn}(\hvx)(\vx - \hvx)$ with
$\vA^{\mathrm{br}}_{mn}(\hvx) \coloneqq \Re\bigl\{(\overline{\hat{i}_{mn}}/|\hat{i}_{mn}|)\,\vD^{\mathrm{br}}_{mn}\bigr\}$,
or, in stacked form,
\begin{align}
	\vb_{\mathrm{br}}(\vx) &\approx \vb_{\mathrm{br}}(\hvx) + \vA_{\mathrm{br}}(\hvx)(\vx - \hvx), \nonumber \\
	\vb_{\mathrm{br}}(\hvx) &\coloneqq |\hat{\vi}^{\mathrm{br}}|, \nonumber \\
	\vA_{\mathrm{br}}(\hvx) &\coloneqq \diag(|\hat{\vi}^{\mathrm{br}}|)^{-1} \Re\left\{\diag(\overline{\hat{\vi}^{\mathrm{br}}})\,\vD^{\mathrm{br}}\right\}. \label{eq:A_br}
\end{align}
These blocks are stacked with the voltage-magnitude Jacobian $\vA_v(\hvx)$ in the main text to form the full linearized constraint model $\vA(\hvx)\vx + \vb(\hvx)$.

The complete block definitions used in the main text are
\begin{align*}
	&\vA(\hvx) = \begin{bmatrix}
		\vA_v(\hvx) \\
		\vA_s(\hvx) \\
		\vA_\mathrm{ang}(\hvx) \\
		\vA_\mathrm{br}(\hvx)
	\end{bmatrix},
	\vb(\hvx) = \begin{bmatrix}
		|\hvv| - \vA_v(\hvx)\hvx \\
		|\hvs_0| - \vA_s(\hvx)\hvx \\
		\vb_{\mathrm{ang}}(\hvx) - \vA_\mathrm{ang}(\hvx)\hvx \\
		|\hat{\vi}^{\mathrm{br}}| - \vA_\mathrm{br}(\hvx)\hvx
	\end{bmatrix}, \\
	&\vl =
	\begin{bmatrix}
		|\vv|^{\mathrm{min}} \\
		-\infty\,\mathbf{1}_{n_{\mathrm{sub}}} \\
		-\infty\,\mathbf{1}_{n_{\mathrm{ang}}} \\
		-\infty\,\mathbf{1}_{n_{\mathrm{br}}}
	\end{bmatrix},
	\vu =
	\begin{bmatrix}
		|\vv|^{\mathrm{max}} \\
		\vs_0^{\mathrm{max}} \\
		\mathbf{0} \\
		\vi^{\mathrm{br,max}}
	\end{bmatrix}
\end{align*}
where $\vA_v$ is given by \eqref{eq:Av_vmag}, and $\vA_s$, $\vA_\mathrm{ang}$, $\vA_\mathrm{br}$, $\vb_\mathrm{ang}$ are given in Appendix~\ref{appendix:lin_constraints} by \eqref{eq:As_sub}, \eqref{eq:A_ang}, \eqref{eq:A_br}, and \eqref{eq:g_ang}.
The row dimensions are $n_{\mathrm{sub}} := \dim(\vs_0^{\mathrm{max}})$, $n_{\mathrm{ang}} := 2|\cE|$, and $n_{\mathrm{br}} := |\cE|$.

\begin{figure}[t]
	\centering
	\includegraphics[width=0.5\linewidth]{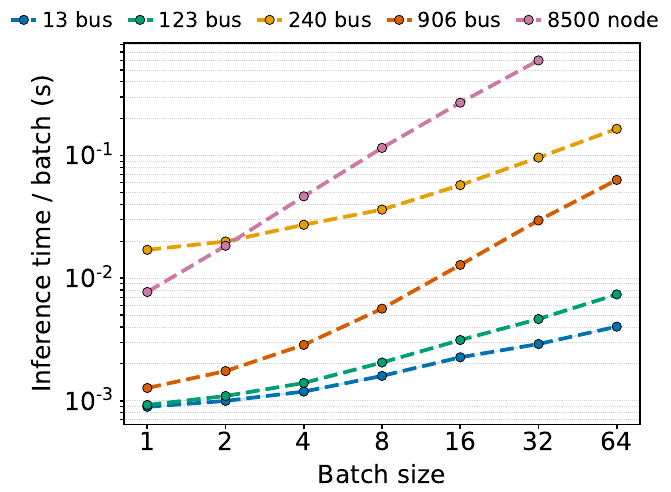}
	\caption{Inference time of the Penalty + SLFS method with different batch sizes.}
	\label{fig:batch_size_inference_time}
\end{figure}

\begin{figure}[t]
	\centering
	\includegraphics[width=0.8\linewidth]{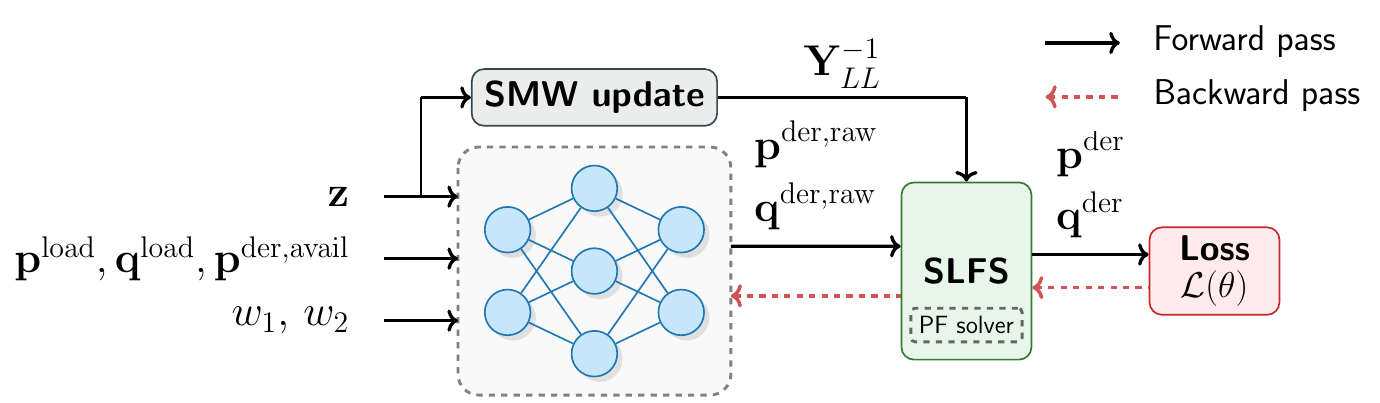}
	\caption{Architecture of the SLFS-FSNet method.}
	\label{fig:fsnet_framework}
\end{figure}

\subsection{Proof of Lemma~\ref{lem:continuity} and Proposition~\ref{prop:approx_error}}

\begin{lemma}[A useful inequality]
\label{lem:inv_square_lipschitz}
Let $c>0$. For any $a,b\in\mathbb C$ satisfying $|a|\ge c$ and $|b|\ge c$,
\[
\left|\frac{1}{a^2}-\frac{1}{b^2}\right|
\le \frac{2}{c^3}\,|a-b|.
\]
\end{lemma}

\begin{proof}
We have
\begin{align*}
\left|\frac{1}{a^2}-\frac{1}{b^2}\right|
&\le |b-a|\frac{|a|+|b|}{|a|^2|b|^2}
= |b-a|\left(\frac{1}{|a||b|^2}+\frac{1}{|a|^2|b|}\right).
\end{align*}
Since $|a|,|b|\ge c$, the bracketed term is at most
\[
\frac{1}{c\cdot c^2}+\frac{1}{c^2\cdot c}
=\frac{2}{c^3}, \qedhere
\]
which proves the claim.
\end{proof}

\noindent\textbf{Proof of Lemma~\ref{lem:continuity}:}
\begin{proof}
    Throughout, write $\vR(\vz)$ for the inverse load-block admittance as in Assumption~\ref{as:operating_domain}.
    \textit{Bound on $\infnorm{J_{\vs}\vG(\vv,\vs)}$}: \\
    With a perturbation $d \vs$ in $\vs$, define
    $d \vG \coloneqq \vG(\vv,\vs+d \vs)-\vG(\vv,\vs)$.    
    By definition, we have $d \vG=\vR(\vz)\diag(\overline{\vv})^{-1}\overline{d \vs}$.
    Using submultiplicativity of induced norms and $\infnorm{\overline{d \vs}}=\infnorm{d \vs}$,
    \begin{align*}
    \infnorm{d \vG}
    &\le \infnorm{\vR(\vz)} \left(\max_i \frac{1}{|v_i|}\right)\infnorm{d \vs} \\
    &\le \frac{\infnorm{\vR(\vz)}}{v^\mathrm{min}} \infnorm{d \vs} =: L_\vs \infnorm{d \vs}.
    \end{align*}
    Taking the supremum over $d \vs$ yields the norm bound
    \begin{align*}
        \infnorm{J_{\vs}\vG(\vv,\vs)} 
        = \sup_{d \vs\neq 0}\frac{\infnorm{d \vG}}{\infnorm{d \vs}} 
        \le \frac{\infnorm{\vR(\vz)}}{v^\mathrm{min}}.
    \end{align*}

    \textit{Bound on $\infnorm{J_{\vv}\vG(\vv,\vs)}$}: \\
    Fixing $\vs$, we have $d\vG = -\vR(\vz)\,\diag(\overline{\vs}) \diag(\overline{\vv})^{-2} d\overline{\vv}$.
    Using $\infnorm{d\overline{\vv}}=\infnorm{d\vv}$, we obtain
    \begin{align*}
    \infnorm{J_{\vv}\vG(\vv,\vs)}
    &\le \infnorm{\vR(\vz)}\infnorm{\diag(\overline{\vs})} \infnorm{\diag(\overline{\vv})^{-2}} \\
    &\le \frac{\infnorm{\vR(\vz)} s^\mathrm{max}}{(v^\mathrm{min})^2}
    =: \beta.
    \end{align*}

    \textit{Lipschitzness of $J_{\vv}\vG(\cdot,\vs)$}: \\
    Let $(\vv_1,\vs),(\vv_2,\vs)\in\Omega$. From the expression for the $\vv$-Jacobian, we have the mapping
    \begin{align*}
        &J_{\vv}\vG(\vv_1,\vs)-J_{\vv}\vG(\vv_2,\vs)
        \;:\; \\
        &\qquad d \vv \mapsto 
        -\vR(\vz)\diag(\overline{\vs})\Big(\diag(\overline{\vv_1})^{-2}-\diag(\overline{\vv_2})^{-2}\Big)\overline{d \vv}.
    \end{align*}
    By submultiplicativity, we have
    \begin{align*}
    &\infnorm{J_{\vv}\vG(\vv_1,\vs)-J_{\vv}\vG(\vv_2,\vs)} \\
    &\le \infnorm{\vR(\vz)}\infnorm{\diag(\overline{\vs})}
    \infnorm{\diag(\overline{\vv_1})^{-2}-\diag(\overline{\vv_2})^{-2}} \\
    &= \infnorm{\vR(\vz)}\infnorm{\vs}
    \max_i \left|\frac{1}{\overline{(v_1)_i}^{2}}-\frac{1}{\overline{(v_2)_i}^{2}}\right|.
    \end{align*}
    Applying Lemma~\ref{lem:inv_square_lipschitz} with
    $a=\overline{(v_1)_i}$ and $b=\overline{(v_2)_i}$ (note that
    $|a|=|(v_1)_i|$ and $|b|=|(v_2)_i|$), and using $\min_i |(v_1)_i|,\min_i |(v_2)_i|\ge v^\mathrm{min}$, we obtain
    \[
    \max_i \left|\frac{1}{\overline{(v_1)_i}^{2}}-\frac{1}{\overline{(v_2)_i}^{2}}\right|
    \le \frac{2}{(v^\mathrm{min})^3}\infnorm{\vv_1-\vv_2}.
    \]
    Therefore, since $\infnorm{\vs}\le s^\mathrm{max}$ on $\Omega$,
    \begin{align*}
        \infnorm{J_{\vv}\vG(\vv_1,\vs)-J_{\vv}\vG(\vv_2,\vs)} \hspace{-1.3cm}&\\
        &\le \frac{2\infnorm{\vR(\vz)}s^\mathrm{max}}{(v^\mathrm{min})^3}\infnorm{\vv_1-\vv_2} \\
        &=: L_\vv\infnorm{\vv_1-\vv_2}. \qedhere
    \end{align*}
\end{proof}

The following lemma is adapted from Lemma 2 in \cite{bolte2023one} to complex matrices, and the proof remains the same.
\begin{lemma}\label{le:matrix_ineq}
    For $\vA \in \C^{n \times n}$, $\infnorm{\vA} < 1$, and $\vB, \tilde{\vB} \in \C^{n \times m}$, then $\|(\vI-\vA)^{-1} \vB - \tilde{\vB}\|_\infty \leq \frac{\infnorm{\vA}}{1-\infnorm{\vA}} \infnorm{\vB} + \|\vB - \tilde{\vB}\|_\infty$.
\end{lemma}

\begin{lemma} \label{le:Jac_M_bounds}
    Under Assumption~\ref{as:operating_domain}, the following bounds hold:
    \begin{align*}
        &\infnorm{J_\vv \Phi_M (\vv_\star(\vs), \vs)} \leq \beta^M, \infnorm{ J_\vs \Phi_M (\vv_\star(\vs), \vs)} \leq \frac{1 - \beta^M}{1 - \beta} L_\vs.
    \end{align*}
\end{lemma}
\begin{proof}
    The first bound follows directly from the chain rule:
    \begin{align*}
        \infnorm{J_\vv \Phi_M (\vv_\star(\vs), \vs)} &= \infnorm{\prod_{m=1}^M J_\vv \vG(\vv_\star(\vs), \vs)} \\
        &\leq \infnorm{J_\vv \vG (\vv_\star(\vs), \vs)}^M \leq \beta^M.
    \end{align*}
    To prove the second bound, we define $\vA_\star = J_\vv \vG(\vv_\star(\vs), \vs)$ and $\vB_\star = J_\vs \vG(\vv_\star(\vs), \vs)$. Using the chain rule and triangle inequality, we have
    \begin{align*}
        \infnorm{J_\vs \Phi_M (\vv_\star(\vs), \vs)}
        &= \left\| \sum_{m=0}^{M-1} \vA_\star^m \vB_\star \right\|_\infty \\
        &\leq \sum_{m=0}^{M-1} \infnorm{\vA_\star}^m \infnorm{\vB_\star} \\
        &\leq \sum_{m=0}^{M-1} \beta^m L_\vs = \frac{1 - \beta^M}{1 - \beta} L_\vs. \qedhere
    \end{align*}
\end{proof}

\begin{lemma} \label{le:product_identity}
    For $r \in \N$ and $\vM_1, \dots, \vM_r, \vN_1, \dots, \vN_r \in \C^{n \times n}$, the following identity holds:
    \begin{align*}
        \prod_{t=1}^r \vM_t - \prod_{t=1}^r \vN_t = \sum_{s=1}^r \left( \prod_{t=s+1}^{r} \vM_t \right) (\vM_s - \vN_s) \left( \prod_{t=1}^{s-1} \vN_t \right).
    \end{align*}
    where the empty product is defined as the identity matrix.
\end{lemma}
\begin{proof}
    We prove the identity by induction on $r$. For $r=1$, both sides equal $\vM_1 - \vN_1$. Suppose the identity holds for $r-1$, then we have
    \begin{align*}
        &\prod_{t=1}^r \vM_t - \prod_{t=1}^r \vN_t = \vM_r \prod_{t=1}^{r-1} \vM_t - \vN_r \prod_{t=1}^{r-1} \vN_t \\
        &= \vM_r \prod_{t=1}^{r-1} \vM_t - \vM_r \prod_{t=1}^{r-1} \vN_t + \vM_r \prod_{t=1}^{r-1} \vN_t - \vN_r \prod_{t=1}^{r-1} \vN_t \\
        &= \vM_r \left( \prod_{t=1}^{r-1} \vM_t - \prod_{t=1}^{r-1} \vN_t \right) + (\vM_r - \vN_r) \prod_{t=1}^{r-1} \vN_t \\
        &= \sum_{s=1}^{r} \left( \prod_{t=s+1}^{r} \vM_t \right) (\vM_s - \vN_s) \left( \prod_{t=1}^{s-1} \vN_t \right). \qedhere
    \end{align*}
\end{proof}

\noindent\textbf{Proof of Proposition~\ref{prop:approx_error}:}
\begin{proof}
    For notational brevity, we define 
    \begin{align*}
        &\vA \coloneqq J_\vv \Phi_M (\vv_\star(\vs), \vs), \quad \vB \coloneqq J_\vs \Phi_M (\vv_\star(\vs), \vs), \\
        &\tilde{\vB} \coloneqq J_\vs \Phi_M (\vv_{k - M}(\vs), \vs) \\
        & \vA_t = J_\vv \vG (\vv_t(\vs), \vs), \quad \vB_t = J_\vs \vG (\vv_t(\vs), \vs), \\
        &\vA_\star = J_\vv \vG(\vv_\star(\vs), \vs), \quad \vB_\star = J_\vs \vG(\vv_\star(\vs), \vs).
    \end{align*}
    By chain rule, we immediately have the following relations:
    \begin{align}
        \vB &= \sum_{m=0}^{M-1} \left(\prod_{t=m+1}^{M-1} \vA_\star \right) \vB_\star, \label{eq:B_relations1}\\
        \tilde{\vB} &= \sum_{m=0}^{M-1} \left( \prod_{t=m+1}^{M-1} \vA_{k - M + t} \right) \vB_{k - M + m} \label{eq:B_relations2}
    \end{align}

    The exact Jacobian at the fixed point is given by
    \begin{align*}
        J_\vs \vv_\star(\vs) &= (\vI - J_\vv \Phi_M(\vv_\star(\vs), \vs))^{-1} J_\vs \Phi_M(\vv_\star(\vs), \vs) \\
        &= (\vI - \vA)^{-1} \vB.
    \end{align*}

    Appealing to Lemma~\ref{le:matrix_ineq}, we have 
    \begin{align} \label{eq:Jac_error_bound}
        \| J^{(M)}_\vs \vv_k(\vs) - J_\vs \vv_\star(\vs)\|_\infty \leq \frac{\infnorm{\vA}}{1-\infnorm{\vA}} \infnorm{\vB} + \|\vB - \tilde{\vB} \|_\infty.
    \end{align}

    The first term in the right-hand side can be bounded using Lemma~\ref{le:Jac_M_bounds} as follows:
    \begin{align} \label{eq:Jac_first_bound}
        \frac{\infnorm{\vA}}{1-\infnorm{\vA}} \infnorm{\vB} \leq \frac{\beta^M}{1 - \beta} L_\vs.
    \end{align}

    To bound the second term, we use \eqref{eq:B_relations1} and \eqref{eq:B_relations2} to have 
    \begin{align*}
        \vB - \tilde{\vB} &= \underbrace{\sum_{m=0}^{M-1} \left( \prod_{t=m+1}^{M-1} \vA_{k - M + t} \right) (\vB_\star - \vB_{k - M + m})}_{\vT_1} \\
        &+ \underbrace{\sum_{m=0}^{M-1} \left( \prod_{t=m+1}^{M-1} \vA_\star - \prod_{t=m+1}^{M-1} \vA_{k - M + t} \right) \vB_\star}_{\vT_2}.
    \end{align*}
    
    For $\vT_1$, we have
    \begin{align}
        \infnorm{\vT_1} &\leq \sum_{m=0}^{M-1} \beta^{M -1 - m} L_\vs \infnorm{\Delta \vv_{k-M+m}(\vs)} \nonumber\\
        &\leq \sum_{j=1}^{M} \beta^{j-1} L_\vs \infnorm{\Delta \vv_{k - j}(\vs)} \label{eq:T1_upper}
    \end{align}
    where the first inequality follows from Lemma~\ref{lem:continuity}, and the second inequality is obtained by changing the index $j = M - m$.

    For $\vT_2$, using Lemma~\ref{le:product_identity}, we have
    \begin{align*}
        &\infnorm{\prod_{t=m+1}^{M-1} \vA_\star - \prod_{t=m+1}^{M-1} \vA_{k - M + t}} = \\
        &\infnorm{\sum_{s=m+1}^{M-1} \left( \prod_{t=s+1}^{M-1} \vA_\star \right) (\vA_\star - \vA_{k - M + s}) \left( \prod_{t=m+1}^{s-1} \vA_{k - M + t} \right)} \\
        &\leq \sum_{s=m+1}^{M-1} \beta^{M- s - 1} \infnorm{\vA_\star - \vA_{k - M + s}} \beta^{s - 1 - m} \\
        &\leq \sum_{s=m+1}^{M-1} \beta^{M - m - 2} L_\vv \infnorm{\Delta\vv_{k - M + s}(\vs)}.
    \end{align*}
    where the first inequality follows from the triangle inequality and sub-multiplicative property of matrix norm, and the second inequality is due to Lemma~\ref{lem:continuity}.

    Note that the bracket under the sum in $\vT_2$ is zero when $m=M-1$. Thus, we multiply the above inequality by $\norm{\vB_\star} \leq L_\vs$ and sum over $m=0, \dots, M-2$ to obtain
    \begin{align}
        \infnorm{\vT_2} &\leq L_\vv L_\vs \sum_{m=0}^{M-2} \beta^{M - m - 2} \sum_{s=m+1}^{M-1} \infnorm{\Delta \vv_{k - M + s}(\vs)} \nonumber\\
        &= L_\vv L_\vs \sum_{s=1}^{M-1} \infnorm{\Delta \vv_{k - M + s}(\vs)} \sum_{m=0}^{s-1} \beta^{M - m - 2} \label{eq:T2_upper_1}
    \end{align}
    where the second line swaps the order of summation. Calculating the inner geometric sum, we have
    \begin{align*}
        \sum_{m=0}^{s-1} \beta^{M - m - 2} = \beta^{M - 2} \frac{1 - \beta^{-s}}{1 - \beta^{-1}} &= \beta^{M-1-s} \frac{1 - \beta^s}{1- \beta} \\
        &\leq \frac{\beta^{M-1 -s}}{1 - \beta}.
    \end{align*}
    Plugging this back into \eqref{eq:T2_upper_1}, we get
    \begin{align}
        \infnorm{\vT_2} &\leq \frac{L_\vv L_\vs}{1 - \beta} \sum_{s=1}^{M-1} \beta^{M - 1 - s} \infnorm{\Delta \vv_{k - M + s}(\vs)} \nonumber\\
        &= \frac{L_\vv L_\vs}{1 - \beta} \sum_{j=1}^{M-1} \beta^{j - 1} \infnorm{\Delta \vv_{k - j}(\vs)} \label{eq:T2_upper_2}
    \end{align}
    where we reindex $j = M - s$ in the last line. 

    Combining \eqref{eq:T1_upper} and \eqref{eq:T2_upper_2}, we have
    \begin{align}\label{eq:Jac_second_bound}
        \infnorm{\vB - \tilde{\vB}} \leq &\left(L_\vs + \frac{L_\vv L_\vs}{1 - \beta} \right) \sum_{j=1}^{M-1} \beta^{j - 1} \infnorm{\Delta \vv_{k - j}(\vs)} \nonumber \\
        & + L_\vs \beta^{M-1} \infnorm{\Delta \vv_{k - M}(\vs)}
    \end{align}
    Finally, plugging \eqref{eq:Jac_first_bound} and \eqref{eq:Jac_second_bound} into \eqref{eq:Jac_error_bound} finishes the proof.
\end{proof}

\subsection{Proof of Theorem~\ref{thm:violation_descent1}}
We first present two auxiliary lemmas that will be used in the proof of the main theorem.
\begin{lemma}\label{lem:tech_1}
    Given $\vl \leq \vu$, the function $\phi(\vy) = \relu{\vy - \vu} + \relu{\vl - \vy}$ is $1$-Lipschitz continuous. That is, for any $\vy, \vy' \in \R^n$,
    \begin{align*}
        \infnorm{\phi(\vy) - \phi(\vy')} \leq \infnorm{\vy - \vy'}
    \end{align*}
\end{lemma}
\begin{proof}
    Let $\phi_j(y_j) = \relu{y_j - u_j} + \relu{l_j - y_j}$. The scalar function $\phi_j$ is $1$-Lipschitz continuous because its derivative (or subgradient) is bounded by $1$.
    
    For any $\vy, \vy' \in \R^n$:
    \begin{align*}
        \infnorm{\phi(\vy) - \phi(\vy')} &= \max_{j=1,\dots,n} |\phi_j(y_j) - \phi_j(y'_j)| \\
        &\leq \max_{j=1,\dots,n} |y_j - y'_j| = \infnorm{\vy - \vy'}.
    \end{align*}
    Thus, $\infnorm{\phi(\vy) - \phi(\vy')} \leq \infnorm{\vy - \vy'}$.
\end{proof}

\begin{lemma}\label{lem:tech_2}
    Under Assumption~\ref{as:linear_model}, we have
    \begin{align*}
        | V(\vx_{k+1}) - V^\text{lin}(\vx_{k+1}; \vx_k) | \leq (\alpha_k)^{1 + \nu} \delta \infnorm{\Delta \vx_k}^{1 + \nu}
    \end{align*}
\end{lemma}
\begin{proof}
    By triangle inequality and Lemma~\ref{lem:tech_1}, we have
    \begin{align*}
        & |V(\vx_{k+1}) - V^\text{lin}(\vx_{k+1}; \vx_k)| \\
        &\leq \infnorm{\vg(\vx_{k+1}) - (\vA(\vx_k) \vx_{k+1} + \vb(\vx_k))}.
    \end{align*}
    Using Assumption~\ref{as:linear_model}, we obtain
    \begin{align*}
        |V(\vx_{k+1}) - V^\text{lin}(\vx_{k+1}; \vx_k)| 
        &\leq \delta \infnorm{\vx_{k+1} - \vx_k}^{1 + \nu} \\
        &= (\alpha_k)^{1 + \nu} \delta \infnorm{\Delta \vx_k}^{1 + \nu}. \qedhere
    \end{align*}
\end{proof}

\begin{lemma}\label{lem:descent_ineq}
Under Assumption~\ref{as:linear_model}, the update rule \eqref{eq:proximal} and \eqref{eq:damping} yields
\begin{align*}
    V(\vx_{k+1}) \leq V(\vx_{k}) - \alpha_k \frac{\eta}{2} \norm{\Delta \vx_k}^2 + (\alpha_k)^{1 + \nu} \delta \infnorm{\Delta \vx_k}^{1 + \nu}.
\end{align*}
\end{lemma}
\begin{proof}
    By the optimality of $\tilde{\vx}_{k+1}$ in \eqref{eq:proximal}, we have
    \begin{align*}
        V^\text{lin}(\tilde{\vx}_{k+1}; \vx_k) + \frac{\eta}{2} \norm{\tilde{\vx}_{k+1} - \vx_k}^2 \leq V^\text{lin}(\vx_k; \vx_k)
    \end{align*}
    Since $V^\text{lin}(\vx_k; \vx_k) = V(\vx_k)$ by Assumption~\ref{as:linear_model}, we get
    \begin{align}\label{eq:opt_ineq}
        V^\text{lin}(\tilde{\vx}_{k+1}; \vx_k) &\leq V(\vx_k) - \frac{\eta}{2} \norm{\tilde{\vx}_{k+1} - \vx_k}^2 \nonumber \\
        &= V(\vx_k) - \frac{\eta}{2} \norm{\Delta \vx_k}^2
    \end{align}

    By the convexity of $V^\text{lin}(\cdot; \vx_k)$, we have
    \begin{align}\label{eq:opt_convex_ineq}
        V^\text{lin}(\vx_{k+1}; \vx_k) &= V^\text{lin}((1-\alpha_k)\vx_k + \alpha_k \tilde{\vx}_{k+1}; \vx_k) \nonumber\\
        &\leq (1 - \alpha_k) V^\text{lin}(\vx_k; \vx_k) + \alpha_k V^\text{lin}(\tilde{\vx}_{k+1}; \vx_k) \nonumber\\
        &\leq V(\vx_k) - \alpha_k \frac{\eta}{2} \norm{\Delta \vx_k}^2
    \end{align}
    where the last inequality uses \eqref{eq:opt_ineq} and $V^\text{lin}(\vx_k; \vx_k) = V(\vx_k)$.

    Using Lemma~\ref{lem:tech_2} and \eqref{eq:opt_convex_ineq}, we have
    \begin{align*}
        V(\vx_{k+1}) &\leq V^\text{lin}(\vx_{k+1}; \vx_k) + (\alpha_k)^{1 + \nu} \delta \infnorm{\Delta \vx_k}^{1 + \nu} \\
        &\leq V(\vx_k) - \alpha_k \frac{\eta}{2} \norm{\Delta \vx_k}^2 + (\alpha_k)^{1 + \nu} \delta \infnorm{\Delta \vx_k}^{1 + \nu} \qedhere
    \end{align*}    
\end{proof}

\noindent\textbf{Proof of Theorem~\ref{thm:violation_descent1}:}
\begin{proof}
    From the choice of $\alpha_k$, we have
    \begin{align*}
        \alpha_k^{1 + \nu} \leq \alpha_k \left(\frac{c \eta}{2 \delta}\right) \norm{\Delta \vx_k}^{1 - \nu}
    \end{align*}
    By Lemma~\ref{lem:descent_ineq} and using $\infnorm{\Delta \vx_k} \leq \norm{\Delta \vx_k}$, we have
    \begin{align*}
        V(\vx_{k+1}) &\leq V(\vx_k) - \alpha_k \frac{\eta}{2} \norm{\Delta \vx_k}^2 \\
        & + \alpha_k \left(\frac{c \eta}{2}\right) \norm{\Delta \vx_k}^{1 - \nu} \norm{\Delta \vx_k}^{1 + \nu} \\
        &= V(\vx_k) - (1 - c) \alpha_k \frac{\eta}{2} \norm{\Delta \vx_k}^2 \qedhere
    \end{align*}
\end{proof}


\end{document}